\documentclass[12pt]{article} 
\usepackage[sectionbib]{natbib}
\usepackage{array,epsfig,fancyheadings,rotating}
\usepackage{amsmath}
\usepackage{mathtools}
\usepackage{setspace}
\usepackage{amssymb}
\usepackage{amsthm}
\usepackage{multirow}
\usepackage{bm}
\usepackage{natbib}
\usepackage{geometry}
\usepackage{subfig}
\usepackage{soul}
\usepackage{color}

\newcommand{\E}{\mathbb{E}}

\newcommand{\floor}[1]{\lfloor #1 \rfloor}
\newcommand{\1}[1]{\boldsymbol{1}\{#1\}}

\newcommand{\N}{\mathbb{N}}

\newcommand{\weak}{\rightsquigarrow }

\usepackage[]{hyperref}  
\usepackage{sectsty, secdot}
\sectionfont{\fontsize{12}{14pt plus.8pt minus .6pt}\selectfont}
\renewcommand{\theequation}{\thesection\arabic{equation}}
\subsectionfont{\fontsize{12}{14pt plus.8pt minus .6pt}\selectfont}

\usepackage{amsmath}
\usepackage{amssymb}
\usepackage{amsfonts}
\usepackage{multirow}
\usepackage{amsthm}

\newtheorem{theorem}{Theorem}

\newtheorem{proposition}{Proposition}
\theoremstyle{definition}

\newtheorem{assumption}{Assumption}

\newcommand{\change}[1]{{\color{black}{#1}}}

\begin{document}


\renewcommand{\baselinestretch}{1}

\markright{ \hbox{\footnotesize\rm Statistica Sinica
}\hfill\\[-13pt]
\hbox{\footnotesize\rm
}\hfill }

\markboth{\hfill{\footnotesize\rm FIRSTNAME1 LASTNAME1 AND FIRSTNAME2 LASTNAME2} \hfill}
{\hfill {\footnotesize\rm FILL IN A SHORT RUNNING TITLE} \hfill}

\renewcommand{\thefootnote}{}
$\ $\par


\fontsize{12}{14pt plus.8pt minus .6pt}\selectfont \vspace{0.8pc}
\centerline{\large\bf Adaptive Change Point Monitoring for High-Dimensional Data }
\vspace{2pt} 
\centerline{Teng Wu, Runmin Wang, Hao Yan, Xiaofeng Shao} 
\vspace{.4cm} 
\centerline{\it Department of
		Statistics, University of Illinois at Urbana-Champaign}
\centerline{\it Department of Statistical Science, Southern Methodist University}
\centerline{ \it School of Computing Informatics \& Decision Systems Engineering, Arizona State University}
 \vspace{.55cm} \fontsize{9}{11.5pt plus.8pt minus.6pt}\selectfont


\begin{quotation}
\noindent {\it Abstract:}
In this paper, we propose a class of monitoring statistics for a mean shift in a sequence of high-dimensional observations. Inspired by the recent U-statistic based retrospective tests developed by \cite{wang2019inference} and \cite{zhangwangshao20}, we advance the U-statistic based approach to the sequential monitoring problem by developing a new  adaptive monitoring procedure that can detect both dense and sparse changes in real time. Unlike \cite{wang2019inference} and \cite{zhangwangshao20}, where self-normalization was used in their tests, we instead introduce a class of estimators for $q$-norm of the covariance matrix and prove their ratio consistency. To facilitate fast computation,  we further develop recursive algorithms to improve the computational efficiency of the monitoring procedure. The advantage of the proposed methodology is demonstrated via simulation studies and real data illustrations.

\vspace{9pt}
\par
\end{quotation}\par

\def\thefigure{\arabic{figure}}
\def\thetable{\arabic{table}}

\renewcommand{\theequation}{\thesection.\arabic{equation}}

\fontsize{12}{14pt plus.8pt minus .6pt}\selectfont

\section{Introduction}

Change point detection problems have been extensively studied in many areas, such as statistics, econometrics and engineering, and there are wide applications to fields in physical science and engineering. The literature is huge and is still growing rapidly. For the low-dimensional data, it dates back to early work by \cite{page1954continuous,macneill1974tests,brown1975techniques}, among others. More recent work that studied change point problems for low/fixed dimensional multivariate time series data can be found in
\cite{shaozhang10,matteson2014nonparametric, kirch2015detection,bucher2019combining}, among others. We refer to \cite{perron2005dealing}, \cite{aue2013structural} and \cite{aminikhanghahi2017survey} for some excellent reviews on this topic. 

\change{
The literature on change point detection can be roughly divided into two categories:  retrospective testing and estimation of change points based on a complete data sequence offline and online sequential monitoring for change points based on some training data and sequentially arrived data. This paper is concerned with the sequential monitoring problem for temporally independent but cross-sectionally dependent high-dimensional data. There are two major lines of research for sequential change-point detection/monitoring. In one line, a huge body of work follows the paradigm set by pioneers in the field, such as \cite{wald1945}, \cite{page1954continuous} and \cite{lorden1971};  see \cite{lai:95, lai:01} and \change{\cite{polunchenko2012state}} for comprehensive reviews. Most sequential detection methods along this line are optimized to have a minimal detection delay with a control of average run length under the null and also the existing procedures are mostly developed for low-dimensional data. These methods often require both
pre-change distribution and post-change distribution to be specified or some
parametric assumption to be made. In the other line,  \cite{chu:96} assumed that there is a stretch of training data  (without any change points) and sequential monitoring was applied to sequentially arrived testing data. They employ the invariance principle to control the type I error and their framework has been adopted by many other researchers in both parametric and nonparametric contexts. 
See 
\cite{horvath2004monitoring,aue2012sequential,wied2013monitoring,fremdt2015page,dette2019likelihood}. Along this line, it is typical to use size and power
(plus average detection delay) to describe and compare operating characteristics
of competing procedures. Our procedure falls into the second category. 
It seems to us that these two frameworks are in general
difficult to compare, as they differ in terms of the model assumptions and evaluation
criteria etc.}

Nowadays, with the rapid improvement of data acquisition technology,  high-dimensional data streams involving continuous sequential observations appear frequently in modern manufacturing and service industries, and the demand for efficient online monitoring tools for such data has never been higher. 
For example, in \cite{yan2018real}, they proposed a method to monitor a multi-channel tonnage profile used for the forging process, which has thousands of attributes. Furthermore, image-based monitoring [\cite{yan2014image}] has become popular in the literature, which includes thousands of pixels per image. In \cite{levy2009detection}, they considered the problem of monitoring thousands of Internet traffic metrics provided by a French Internet service provider. This kind of high-dimensional data poses significant new challenges to the traditional multivariate statistical process control and monitoring, since when the dimension $p$
is high and is comparable to the sample size $n$, most existing sequential monitoring methods constructed based on the fixed-dimension assumptions become invalid. 

In this article, we propose a new class of sequential monitoring methodology to monitor a change in the mean of independent high-dimensional data based on (sequential) retrospective testing. Our proposal is inspired by recent work on retrospective testing of mean change in high-dimensional data by \cite{wang2019inference} and \cite{zhangwangshao20}. In  \cite{wang2019inference}, the authors proposed a U-statistic based approach to target the $L_2$-norm of the mean difference by extending the U-statistic idea initiated in \cite{chenqin10} from two-sample testing to the change point testing problem. In \cite{zhangwangshao20}, they further extended the test in \cite{wang2019inference} to a $L_q$-norm-based one mimicking \cite{he2018asymptotically}, where $q\in 2\N$, to capture sparse alternative. By combining $L_2$-norm-based test and $L_q$-norm-based one, the adaptive test statistic they proposed is shown to achieve high power for both dense and sparse alternatives. However, one of the limitations of these works is that these methods are designed for off-line analysis, which is not suitable to be applied to real-time online monitoring systems. Building on the works of \cite{wang2019inference} and \cite{zhangwangshao20}, we shall propose a new adaptive sequential monitoring procedure that can capture both sparse and dense alternatives.  Instead of using the self-normalization scheme [\cite{shao10,shaozhang10,shao15}],  as done in \cite{wang2019inference} and \cite{zhangwangshao20}, we opt to use ratio-consistent estimators for $\|\Sigma\|_q^q$ based on the training data, where $\Sigma$ is the common covariance matrix of the sequence of random vectors and provide a rigorous proof. Further, we develop recursive algorithms for fast implementation so that at each time the monitoring statistics can be efficiently computed. The resulting adaptive monitoring procedure via a combination of $L_2$ and $L_q$ (say $q =6$) based sequential tests are shown to be powerful against both dense and sparse alternatives via theory and simulations.

There is a growing literature on high-dimensional change point detection in the retrospective setting; see \cite{horvath2012change,cho2015multiple,jirak2015uniform,yu2017finite,wang2018high,  yu2019robust, wang2019inference, zhangwangshao20, wangshao20}, among others. It is worth noting that \cite{enikeeva2019high} developed a test based on a combination of a linear statistic
and a scan statistic, and their test can be adaptive to both sparse and dense alternatives. However, their  Gaussian and independent components assumptions are also too restrictive. In addition,  the literature for online monitoring of high-dimensional data streams has also been growing steadily in the literature of statistics and quality control in the last decade. In particular, \cite{mei:10} proposed a global monitoring scheme based on the sum of the cumulative sum monitoring statistics from each individual
data stream. His method aims to minimize the delay time and control the global false alarm rate,  which is based on the average run length under the null. This is different from the size and power analysis as done in our work. Note that the assumptions in \cite{mei:10} are quite restrictive in the sense that he assumed all data streams \change{do not have cross-sectional dependence}, and that both the pre-change and post-change distributions are known. See \cite{wang:15},   \cite{zou2015efficient}, \cite{liu:19}, and \cite{li2020} for several variants in the sense that they propose new ways of aggregating the local monitoring statistics. In 
\cite{xie2013sequential}, they proposed a mixture detection procedure based on a likelihood ratio statistic that takes into account the fraction of data streams being affected. They argue that the performance is good when the fraction of affected data streams are known and do not require to completely specify the post-change distribution. However, the mixture global log-likelihood they specified relies on that hypothesized affected fraction $p_0$, and they showed the robustness of different choices of $p_0$ only through numerical studies. The results they derived hold for data generated from a normal distribution or other exponential families of distributions. A common feature of all these works is that they assume the data streams \change{do not have cross-sectional dependence}, which may be violated in practice.  As a matter of fact, our theory for the proposed monitoring statistic demonstrates the impact of the correlation/covariance structure of the multiple data streams, which seems not well appreciated in the above-mentioned literature.

The rest of the paper is structured as follows: In Section 2, we specify the change point monitoring framework we use and propose the monitoring statistic that targets the $L_q$-norm of the mean change. An adaptive monitoring scheme can be derived by combining the test statistic for different $q$'s, $q\in 2\N$. Section 3 provides a ratio-consistent estimator for $\|\Sigma\|_q^q$, which is crucial when constructing the monitoring statistics. Section 4 provides simulation studies to examine the finite sample performance of the adaptive monitoring statistic. In Section 5, we apply the adaptive monitoring scheme to two real datasets. Section 6 concludes the paper. All the technical details can be found in the Appendix.

\section{Monitoring Statistics}
In this section, we specify the general framework we use to perform change point monitoring. We consider a closed-end change point monitoring scenario following \cite{chu:96}. Assume that we observe a sequence of temporally independent high dimensional observations $X_1, \ldots, X_n \in \mathbb R^p$, which are ordered in time and have constant mean $\bm \mu$ and covariance matrix $\Sigma$. We start the monitoring procedure from time $(n+1)$ to detect if the mean vector changes in the future. Throughout the analysis, we assume that all $X_t$'s are independent over time. A decision is to be made at each of the time points, and we will signal an alarm when the monitoring statistic exceeds a certain boundary. The process ends at time $nT$ regardless of whether a change point is detected, where $T$ is a pre-specified number. The Type-I error of the monitoring procedure is controlled at $\alpha$, which means the probability of signaling an alarm when there is no change within the period $[n+1, nT]$ is at most $\alpha$.

Under the null hypothesis, no change occurs within the monitoring period and we have 
\[ E(X_t) = \mu \text{ for $t = 1, \ldots, nT$}.\]
against the alternative 
\[\]
Under the alternatives, the mean function changes at some time $t_0 > n$, and remains at the same level for the following observations. That is
    \[E(X_t) = \begin{cases} \mu &  1< t < t_0\\
    \mu + \Delta & t_0\le t \le nT.\end{cases}\]
    
We propose a family of test statistics $T_{n,q}(k)$, which serves as the monitoring statistic targeting $\|\Delta\|_q$. The case $q = 2$ corresponds to dense alternatives, and larger values of $q$'s correspond to sparser alternatives. We will discuss the formulation of our monitoring statistic for $q = 2$ and then extend to general q's in the following subsections. 
   
\subsection{$L_2$-norm-based monitoring statistics}
In this section, we will first develop the $L_2$-norm-based monitoring statistic, which is especially useful to detect the dense alternative. Furthermore, we will discuss the asymptotic properties of the $L_2$-norm-based statistic. Finally, the recursive computational algorithm  will be developed to allow efficient implementation. 

\subsubsection{Monitoring statistics}
For a given time $k > n$, suppose we know a change point happens at the location $m$, where $n < m <k$.  We can separate the observations into two independent samples: pre-break $X_1,\ldots,X_{m}$ and post-break $X_{m+1},\ldots,X_{k}$. 
Consider using a two-sample U-statistic with kernel 
\[h((X,Y),(X',Y')) = (X- Y)^T(X'-Y')\],
where $(X', Y')$ is an independent copy of $(X, Y)$. Then we have \[E[h((X,Y),(X',Y'))] = \|E(X) - E(Y)\|_2^2, \]
which estimates the squared $L_2$-norm of the mean difference. Indeed \cite{wang2019inference} constructed a $L_2$-norm-based retrospective change point detection statistic by scanning over all possible $m$.  For the online monitoring problem, we shall combine this idea with the approach in \cite{dette2019likelihood} to propose a monitoring statistic. To be more precise, at each time point $k$, we scan through all possible change point locations $m$ ($n< m < k - 2$), and perform a change point testing. We take the maximum of these U-statistics over $m$ as our test statistics at time $k$. Suppose we can get a ratio-consistent estimator of $||\Sigma||_F$ learned from the training sample $\{X_1, \ldots, X_n\}$ denoted by $\widehat{\| {\Sigma}\|_F}$ , our monitoring statistic at time $k = n+3,\ldots, nT$ is
\begin{align*}
T_{n,2}(k) &= \frac{1}{n^3\widehat{\| {\Sigma}\|_F}}\max_{m = n+1\ldots,k-2 }\sum_{l = 1}^p\sum_{1\le i_1,i_2\le m}^*  \sum_{m+1 \le j_1,j_2\le k}^*  (X_{i_1, l} - X_{j_1, l})(X_{i_2, l} - X_{j_2, l})\\
& =  \frac{1}{n^3\widehat{\| {\Sigma}\|_F}}\max_{m = n+1\ldots,k-2 }G_k(m).
\end{align*}

\subsubsection{Asymptotic properties}
To calibrate the size of the testing procedure, we need to obtain the asymptotic distribution of the test statistic under the null. The following conditions are imposed in \cite{wang2019inference} to ensure the process convergence results.
\begin{assumption}
 $tr(\Sigma^4) =o(\|\Sigma\|^4_F)$.
\end{assumption}
\begin{assumption}
Let $Cum(h) = \sum_{l_1,\ldots,l_h=1}^p cum^2(X_{1,l_1},\ldots,X_{1,l_h}) \le C||\Sigma||^h_F$ for $h=2,3,4,5,6$ and some constant $C$.
\change{Here $cum(\cdot)$ is the joint cumulant. In general, for a sequence of random variable $Y_1,\ldots, Y_n$, their joint cumulant is defined as 
  \[cum(Y_1,\ldots,Y_n) = \sum_\pi (|\pi| - 1)!(-1)^{|\pi| - 1} \prod_{B\in \pi}E\left( \prod_{i \in B} Y_i\right),\]
  where $\pi$ runs through the list of all partitions of $\{1,\ldots,n\}$, $B$ runs through the list of all blocks of partition $\pi$ and  and $\pi$ is the number of parts in the partition.}
\end{assumption}
 Assumption 1 was also imposed in \cite{chenqin10}, who pioneered the use of $U$-statistic approach in the two-sample testing problem for high-dimensional data, and it   can be satisfied by a wide range of covariance models. Assumption 2 can be viewed as some restrictions on the dependence structure, which holds under uniform bounds on moments and `short-range’ dependence type conditions on the entries of the vector $(X_{0,1},...,X_{0,p})$. See \cite{wang2019inference} for more discussions about these two assumptions. Finally, under the null hypothesis and these assumptions, we provide the limiting distribution of the proposed monitoring statistic in Theorem \ref{thm: limiting}. 
\begin{theorem}
\label{thm: limiting}
Under Assumptions 1 and 2, we have
\[\max_{k = n+3}^{nT} T_{n,2}(k) \xrightarrow{D}\sup_{t \in [1,T]}\sup_{s\in[1,t]} G(s,t),\]
where
\[G(s,t)  = t(t-s)Q(0,s) + stQ(s,t) - s(t-s)Q(0,t),\]
and $Q$ is a Gaussian process whose covariance structure is the following
\begin{equation*}
    Cov(Q(a_1,b_1),Q(a_2,b_2)) = \begin{cases}
    (\min(b_1,b_2) - \max(a_1,a_2))^2 & if\quad \max(a_1,a_2) \le \min(b_1,b_2)\\
    0 & otherwise
    \end{cases}
\end{equation*} 
\end{theorem}

In general, we can also consider some non-constant boundary function $w(t)$, that is,
\[\max_{k = n+1}^{nT} \frac{T_{n,2}(k)}{w(k/n -1)} \xrightarrow{D}\sup_{t \in [1,T]}\sup_{s\in[1,t]} \frac{G(s,t)}{w(t-1)}.\]
We take the double supremums here to control the familywise error rate. Therefore, we reject the null hypothesis if $T_{n,2}(k) > c_\alpha w(k/n-1)$ for some $k \in \{n+3, \ldots, nT\}$.  The size can be calibrated by choosing a $c_\alpha$, such that
\[P\left(\sup_{t \in [1,T]}\sup_{s\in[1,t]} \frac{G(s,t)}{ w(t-1)} > c_\alpha\right) = \alpha.\]
Different choices of $w(t)$ have been considered in \cite{dette2019likelihood}.
\begin{itemize}
    \item (T1) $w(t) = 1$,
    \item (T2) $w(t) = (t+1)^2$,
    \item (T3) $w(t) = (t+1)^2\cdot \max\Big\{\Big (\frac{t}{t+1}\Big )^{1/2},10^{-10}\Big\}$.
\end{itemize}
\change{These $w(t)$s are
motivated by the law of iterated logarithm and are used to reduce the stopping
delay under the alternative. Based on our simulation results and real data applications,
the choice of $w(t)$ among the above three candidates  does not seem
to have a big impact on power and detection delay. So in practice, for closed-end
procedure, any choice would work. The detailed comparisons are shown in the simulation studies in Section 4.} 

\change{\textbf{Remark}
The current method can be generalized to the open-end framework. For an open-end monitoring procedure, we are interested in testing 
\[ E(X_t) = \mu \text{ for $t = 1, 2,\ldots$}.\]
against the alternative 
 \[E(X_t) = \begin{cases} \mu &  1< t < t_0\\
    \mu + \Delta & t > t_0.\end{cases}\]
for some $t_0 > n$. Suppose we use the same $L_2$ norm based monitoring statistic at time $k = n+3, \ldots$, i.e.,
\[
T_{n,2}(k) =  \frac{1}{n^3\widehat{\| {\Sigma}\|_F}}\max_{m = n+1\ldots,k-2 }G_k(m).
\]
For a suitably chosen boundary function $w(\cdot)$, we expect that
\[\sup_{k =n+3}^\infty \frac{T_{n,2}(k)}{w(k/n -1)} \xrightarrow{D}\sup_{ t \in [1,\infty)}\sup_{s\in[1,t]} \frac{G(s,t)}{w(t-1)},\]
as $n \to \infty$. The critical value can be determined by 
\[P\left(\sup_{t \in [1,\infty)}\sup_{s\in[1,t]} \frac{G(s,t)}{ w(t-1)} > c_\alpha\right) = \alpha.\]
We reject the null hypothesis if $T_{n,2}(k) > c_\alpha w(k/n-1)$ for some $k \in \{n+3, \ldots\}$. In practice, we can approximate critical values $c_\alpha$ based on the procedure we used for simulating the critical values in the closed-end procedure, by using a large $T$, say $T=200$. Note that the boundary function used for open-end monitoring needs to satisfy certain smoothness and decay rate assumptions and the above three we used for the closed-end procedure
are no longer applicable;  see Assumption 2.4 in \cite{gkd2020JTSA} and related discussions.
}

The following theorem provides theoretical analysis for the power of the $L_2$-norm-based monitoring procedure. 
\begin{theorem}
\label{thm: thm2}
Suppose that Assumptions 1 and 2 hold. Further assume that the change point location is at $\floor{nr}$ for some $r \in (1,T)$, we have 
\begin{enumerate}
    \item When $\frac{n\Delta^T\Delta}{||\Sigma||_F} \to 0$, \[\max_{k = n+3,\ldots nT} T_{n,2}(k) \xrightarrow{D}\sup_{t \in [1,T]}\sup_{s\in[1,t]} G(s,t).\]
    \item\change{ When $\frac{n\Delta^T\Delta}{||\Sigma||_F} \to b \in (0 +\infty)$,
    \[\max_{k = n+3,\ldots nT} T_{n,2}(k) \xrightarrow{D}\Tilde{T}_2 = \sup_{t \in [1,T]}\sup_{s\in[1,t]} \left[G(s,t) + b\Lambda(s,t)\right],\]
    where 
   \[\Lambda(s,t) = \begin{cases} (t-r)^2s^2 & s \le r\\
r^2 (t-s)^2 & s > r\\
0 & otherwise
\end{cases}.\]}
    \item When $\frac{n\Delta^T\Delta}{||\Sigma||_F} \to \infty$, \[\max_{k = n+3,\ldots nT} T_{n,2}(k)\xrightarrow{D} \infty.\]
\end{enumerate}
\end{theorem}
Theorem \ref{thm: thm2} implies that, under the local alternative where $\frac{n\Delta^T\Delta}{||\Sigma||_F} \to 0$, the proposed monitoring procedure has trivial power. For the diverging alternative where $\frac{n\Delta^T\Delta}{||\Sigma||_F} \to +\infty$, the test has power converging to 1. When the strength corresponding to the change falls in between, the test has power between $(\alpha, 1)$. 

\subsubsection{Recursive computation}
One challenge for the proposed monitoring statistic $T_{n,2}(k)$ is that it needs to be recomputed at each given time $k$. The brute force calculation of the test statistics has $O(n^4p)$ time complexity and O(1) space complexity. In this section, we develop a recursive algorithm to efficiently update the monitoring statistic, which greatly improves the computational efficiency for online monitoring. More specifically, we propose a recursive algorithm to update $G_k(m)$, which is a major component to compute the monitoring statistic $T_{n,2}(k)$ as follows:
 
\begin{align*}
    G_k(m) &= 
    (k-m)(k-m-1)\sum_{1\le i< j \le m}X_i^TX_j +m(m-1) \sum_{m+1 \le i< j\le k}X_{i}^TX_{j}\\
    & -(m-1)(k-m-1)\sum_{i=1}^{m}\sum_{j=m+1}^{k} X_i^TX_{j}.
\end{align*}
To compute $G_k(m)$, we need to keep track of two CUSUM processes
\[B_t = \sum_{i = 1}^t X_i \text{ and } C_t = \sum_{i = 1}^t X_i^T X_i,\]
where $B_t$'s are still  $p$-dimensional. The partial sum process $S(a,b) = \sum_{a \le i < j \le b} X_i^T X_j$ appeared in $G_k(m)$ can be expressed in terms of functions of $B_t$ and $C_t$,
\[S(a,b) = \sum_{a \le i < j \le b} X_i^T X_j = \frac{1}{2}[(B_b -  B_{a-1})^T(B_b - B_{a-1}) - (C_b - C_{a-1})].\]

The detailed algorithm is stated as follows, 
\begin{enumerate}
    \item \textbf{Initialization}:
    Start with the first pair $(m,k) = (n+1, n+3)$. Record the following quantities
    \[B_{n+1}, B_{n+2}, B_{n+3}, C_{n+1}, C_{n+2}, C_{n+3}.\]
    The first statistic can be calculated based on  
    \begin{align*}
    G_{n+3}(n+1) &= 
 2\cdot (B_{n+1}^T B_{n+1} - C_{n+1})/2 \\
 &+ (n+1)n[(B_{n+3} -  B_{n+1})^T(B_{n+3} -  B_{n+1}) \\
    & - (C_{n+3} - C_{n+1})]/2 - nB_{n+1}^T(B_{n+3} - B_{n+1}).
\end{align*}
  \item \textbf{Increase index from $k$ to $k+1$}:
    Fix index $m$, compute $B_{k+1}$ and $C_{k+1}$:
    \[B_{k+1} = B_k + X_{k+1}, C_{k+1} = C_k + X_{k+1}^T X_{k+1}.\]
    The statistic for the pair $(m, k + 1)$ is 
    \begin{align*}
    G_{k+1}(m) &= 
   (k-m+1)(k-m)(B_{m}^T B_{m} - C_{m})/2\\
   & + m(m-1)[(B_{k+1} -  B_{m})^T(B_{k+1} -  B_{m}))\\
    &- (C_{k+1} - C_{m})]/2 - (m-1)(k-m)\sum_{i=1}^{m} B_{m}^T(B_{k+1} - B_m).
    \end{align*}
    \item \textbf{Increase index from $m$ to $m+1$}:
    For fixed index $k$, all $B_i$ and $C_i$ for $i = n \ldots, k$ are already recorded. 
    The statistic for pair $(m+1, k)$ is 
    \begin{align*}
    G_k(m + 1) &= 
    (k-m-1)(k-m-2)(B_{m+1}^T B_{m+1} - C_{m+1})/2\\
    & + (m+1)m[(B_{k} -  B_{m+1})^T(B_{k} -  B_{m+1})) - (C_{k} - C_{m+1})]/2\\
    & - (k-m-2)mB_{m+1}^T(B_{k} - B_{m+1}).
    \end{align*}
\end{enumerate}
The algorithm should start with $(m,k) = (n+1, n+3)$, increase the second index $k$ first and then increase along the first index $m$. The recursive formulation reduces the time complexity to $O(n^2p)$ with additional space complexity $O(np)$.


\subsection{$L_q$-norm-based monitoring statistics}

In this section, we  generalize the monitoring statistic from $L_2$-norm to $L_q$-norm. As has been shown in the previous analysis, the power of the $L_2$-norm-based monitoring statistic depends on quantity $\|\Delta\|_2$, which is sensitive to dense alternatives. However, in real applications, we usually do not know a priori if the mean change is dense or not. As an approximation, we can consider a similar test statistic targeting $\|\Delta\|_q$, for $q \in 2\mathbb N$. When $q$ is large, we are essentially testing against sparse alternatives. As a special case, if we let $q \to \infty$, $\lim_{q\to \infty} \|\Delta\|_q = \|\Delta\|_\infty$, we  only target on the largest element (in absolute value) of the $\Delta$. 

\subsubsection{Monitoring statistics}

To define the monitoring statistics,  we adopt the idea used in Zhang et al. (2020) without applying self-normalization. Self-normalization requires more extensive computation and can be avoided by using the Phase I data to obtain a ratio consistent estimator for $\|\Sigma\|_q$. Also, as pointed out by \cite{shao15}, self-normalization can result in a slight loss of power. Essentially, we can construct a $L_q$-norm-based test statistic at time $k = n+q + 1, \ldots, nT$,
\begin{align*}
T_{n,q}(k) &= \frac{1}{\sqrt{n^{3q}\widehat{\| \Sigma\|_q^q}}}\max_{m  = n+1,\ldots, k - q}\sum_{l = 1}^p\sum_{1\le i_1, \ldots,i_q\le m}^*  \sum_{m+1 \le j_1, \ldots,j_q\le k}^*  (X_{i_1, l} - X_{j_1, l}) \cdots (X_{i_q, l} - X_{j_q, l})\\
 & = \frac{1}{\sqrt{n^{3q}\widehat{\| \Sigma\|_q^q}}}\max_{m  = n+1,\ldots, k - q}U_{n,q}(k,m),
\end{align*}
where $\widehat{\| \Sigma\|_q^q}$ is a ratio-consistent estimator of $\|\Sigma\|_q^q$.

\subsubsection{Asymptotic properties}

In this subsection, we  study the asymptotic property of the $L_q$-norm-based test statistics. First, we  impose the following conditions in Zhang et al. (2020) to facilitate the asymptotic analysis. 
\begin{assumption}
Let $X_t = \mu + Z_t$. Suppose $Z_t$ are i.i.d copies of $Z_0$ with mean 0 and covariance  matrix $\Sigma$. There exists $c_0 > 0$ independent of $n$ such that $\inf_{i = 1\ldots ,p} Var(Z_{0,i}) \ge c_0$
\end{assumption}
\begin{assumption}
$Z_0$ has up to 8-th moments, with $\sup_{1\le j\le p} E[Z^8_{0,j} ] \le  C$, and for $h = 2\ldots 8$ there exist constants $C_h$
depending on $h$ only and a constant $r > 2$ such that
\[|cum(Z_{0,l_1} \ldots, Z_{0,l_h})| \le C(1 \lor \max_{1 \le i < j\le h} |l_i - l_j|)^{-r}.\]
\end{assumption}

These assumptions appeared in Zhang et al. (2020), and \cite{wang2019inference} showed that they imply Assumptions 1 and 2 for the case $q=2$. Assumption 4 can be implied by geometric moment contraction [cf. Proposition 2 of \cite{wu2004limit}] or physical dependence measure proposed by \cite{wu2005nonlinear} [cf. Section 4 of \cite{shaowu2007}], or $\alpha$-mixing. It essentially requires weak \change{cross-sectional} dependence among the $p$ components in the data. 

Under the null hypothesis, to obtain the limiting distribution of monitoring statistic ${T}_{n,q}$, we utilize the limiting process in \cite{zhangwangshao20}  and obtained the following theorem.
\begin{theorem}
Under Assumptions 3 and 4, \[\max_{k = n+q+1}^{nT}T_{n,q}(k) \xrightarrow{d} \Tilde{T}_q : =  \sup_{t \in [1, T]}\sup_{s \in [1, t]} G_q(s,t),\]
where \[G_q (s,t)= \sum_{c = 0}^q (-1)^{q-c}\begin{pmatrix} q\\ c \end{pmatrix} s^{q-c}(t-s)^cQ_{q,c}(s;[0,t]),\]
and $Q_{q,c}(r;[a,b])$ is a Gaussian process with covariance structure 
\[cov(Q_{q,c_1}(r_1;[a_1,b_1]), Q_{q,c_2}(r_2;[a_2,b_2])) = \begin{pmatrix} C\\ c \end{pmatrix} c!(q-c)!(r-A)^c(R-r)^{C-c}(b-R)^{q-C}, \]
where $A = \max(a_1, a_2)$, $c = \min(c_1,c_2)$, $C = \max(c_1, c_2)$ and $b = \min(b_1, b_2)$. Two processes $Q_{q_1, c_1}$ and $Q_{q_2,c_2}$ are mutually independent if $q_1 \ne q_2 \in 2\N$.
\end{theorem}
 The limiting null distribution is pivotal and its critical values can be simulated based on the following equation,  \[P\left(\sup_{t \in [1,T]}\sup_{s\in[1,t]} \frac{G_q(s,t)}{w(t-1)}> c_\alpha\right) = \alpha.\]
We reject the $H_0$ when ${T}_{n,q}(k) > c_\alpha w(k/n-1)$ for $k = n+q + 1, \ldots, nT$. 
We tabulate the critical values for $T=2$, $q = 2, 6$ and different boundary functions in Table \ref{tab:cv}. Critical values under other settings are available upon request. 

\begin{table}[]
\caption{Simulated critical values for $L_q$-norm-based test, $T=2$}
\centering
\label{tab:cv}
\begin{tabular}{lrrrrrr}
\hline
Boundary & \multicolumn{2}{c}{T1}                                & \multicolumn{2}{c}{T2}                                & \multicolumn{2}{c}{T3}                                \\ \cline{2-7} 
Quantile & \multicolumn{1}{c}{$L_2$} & \multicolumn{1}{c}{$L_6$} & \multicolumn{1}{c}{$L_2$} & \multicolumn{1}{c}{$L_6$} & \multicolumn{1}{c}{$L_2$} & \multicolumn{1}{c}{$L_6$} \\ \hline
90\%     & 0.756                     & 3.235                     & 0.204                     & 0.867                     & 0.141                     & 0.592                     \\
95\%     & 1.264                     & 3.711                     & 0.331                     & 0.973                     & 0.232                     & 0.676                     \\
99\%     & 2.715                     & 4.635                     & 0.706                     & 1.196                     & 0.485                     & 0.837                     \\ \hline
\end{tabular}
\end{table}

Finally, we  study the power of the $L_q$-norm-based monitoring procedure in Theorem \ref{thm: powerlq}. 
\begin{theorem}
\label{thm: powerlq}
Suppose that Assumptions 3 and 4 hold and the change point location is at $\floor{nr}$ for some $r \in (1,T)$, 
\begin{enumerate}
    \item When $\frac{n^{q/2}\|\Delta\|_q^q}{\|\Sigma\|_q^{q/2}} \to 0$, $\max_{k = n+q+1}^{nT}T_{n,q}(k) \xrightarrow{D} \Tilde{T}_q$;
    \item When $\frac{n^{q/2}\|\Delta\|_q^q}{\|\Sigma\|_q^{q/2}} \to \gamma \in (0 , +\infty)$, 
    \[\max_{k = n+q+1}^{nT}T_{n,q}(k) \xrightarrow{D} \sup_{t \in [1,T]}\sup_{s\in[1,t]} [G_q(s,t) + \gamma J_q(s;[0,t])],\]
where 
 \[ J_q(s;[0,t]) = \begin{cases} 
r^q(t-s)^q & r\le s < t\\
s^q(t-r)^q & s \le r  < t\\
0 & otherwise
\end{cases};
\]
    \item When $\frac{n^{q/2}\|\Delta\|_q^q}{\|\Sigma\|_q^{q/2}} \to \infty$, $\max_{k = n+q+1}^{nT}T_{n,q}(k) \xrightarrow{D} \infty$.
\end{enumerate}
\end{theorem}
Analogous to the $q=2$ case, the power of the test depends on $\|\Delta\|_q$. Therefore, for large $q$, the proposed test is sensitive to sparse alternatives. 

\subsubsection{Recursive computation}
Similar to the $L_2$-based-test statistics, we would like to extend the recursive formulation to $L_q$-norm-based test statistic. According to Zhang et al.(2020), under the null, the process $U_{n,q}(k,m)$ can be simplified as 
\[U_{n,q}(k,m) = \sum_{c = 0}^q (-1)^{q-c}\begin{pmatrix} q\\ c \end{pmatrix} P^{m-1-c}_{q-c}P^{k-m-q+c}_{c}S_{n,q,c}(m;1,k), \]
where $P_l^k = k!/(k-l)!$ and
\[S_{n,q,c}(m;s,k) = \sum_{l = 1}^p\sum_{s\le i_1, \ldots,i_c\le m}^*  \sum_{m+1 \le j_1, \ldots,j_{q-c}\le k}^*  \prod_{t=1}^cX_{i_t, l}\prod_{g = 1}^{q-c}X_{j_g, l}.\]

Since $S_{n,q,c}(m;1,k)$'s are the major building blocks of our final test statistic and need to be computed at each time $k$, we need to find efficient ways to calculate them recursively. One key element is the sum of product terms like \[B(c,m,l) : = \sum_{1\le i_1, \ldots,i_c\le m}^* \prod_{t=1}^cX_{i_t, l}, \qquad \text{ and }\]
\[M(c,m,k,l) :=\sum_{m \le j_1, \ldots,j_{c}\le k}^*\prod_{g = 1}^{c}X_{j_g, l}.\]
When we increase from $m$ to $m+1$, 
\[\sum_{1\le i_1, \ldots,i_c\le m+1}^* \prod_{t=1}^cX_{i_t, l} = \sum_{1\le i_1, \ldots,i_c\le m}^* \prod_{t=1}^cX_{i_t, l} + X_{m+1, l} \cdot \sum_{1\le i_1, \ldots,i_{c-1}\le m}^* \prod_{t=1}^{c-1}X_{i_t, l}.\]
We can derive the following recursive relationship for $B(c,k,l)$, i.e. 
\begin{equation}
B(c,m+1,l) = B(c,m,l) + B(c-1,m,l)\cdot X_{k+1, l}.\label{recur}
\end{equation}
There is similar recursive relationship for $M(c,m,k,l)$, 
\begin{equation}
M(c,m+1,k,l) = M(c,m,k,l) + X_{m+1,l}M(c-1, m,k,l).\label{recur2}
\end{equation}

To enable the recursive computation, for each $c = 0, \ldots, q$, we maintain a matrix to store the cumulative sums. 
\begin{enumerate}
  \item \textbf{Initialization}: Starting with $c = 0$ and $c = 1$, for all $l = 1, \ldots, p$,  we initialize $B(0, k+1,l), \ldots, B(0,k+q, l) = 0$ and calculate
\[B(1, k+1,l) = \sum_{i = 1}^{k+1}X_{i,l}, \ldots, B(1,k+q, l)= \sum_{i = 1}^{k+q}X_{i,l}.\]
Then we can recursively calculate $B(c, i,l)$ for all $c = 0, \ldots, q$ and $i \le k+q$ based on Equation \ref{recur}.
  \item \textbf{Update  from $B(c,k,l)$ to $B(c,k+1,l)$}: We let $B(0,k+1,l) = B(0,k,l) + X_{k+1,l}$ and obtain the result for other $B(c,k+1,l)$ $(c \le q)$ using Equation \ref{recur}.
  \item \textbf{Update from $M(c,m,k,l )$ to $M(c,m+1,k,l )$}: Fix index $k$, for any $ n  + 1\le m \le k-q$, $l = 1, \ldots, p$, we let $M(0,m,k,l)= 0$  and calculate 
  \[M(1,m,k,l) = \sum_{i = m}^kX_{i,l}.\]
 All other $M(c,m,k,l )$ where $c \le q$ and $ n  + 1\le m \le k-q$, can be obtained via Equation \ref{recur2}. Construct the test statistic $T_{n,q}(k+1)$ using $B(c,k,l)$'s and  $M(c,m,k,l )$'s and repeat from step 2. 
  \end{enumerate}
  The time complexity of the recursive formulation is $O(n^2pq)$ with space complexity $O(npq)$.

\subsection{Combining multiple $L_q$-norm-based statistics}
In this section, we propose to combine multiple $L_q$ statistics to detect both dense and sparse alternatives. 
More specifically, based on theoretical results in Zhang et al. (2020), the U-process for different $q$'s are asymptotically independent, which implies that $\{T_{n,q}\}_{k = n+q +1}^{nT}$ are asymptotically independent for $q \in 2\N$. Therefore, $\max_{k = n+q + 1}^{nT} T_{n,q}(k)$ are asymptotically independent for $q\in I$, where $I \subset 2\N$, say $I = \{2,6\}$. Thus we can combine the monitoring procedure for different $q$'s and adjust for the asymptotic size. In general, if we want to combine a set of $q \in I$, we can adjust the size of each individual test to be $1- (1-\alpha)^{1/|I|}$ given the asymptotic independence and reject the null if any of the monitoring statistics exceeds its critical value. In Zhang et al.(2020) they provide a discussion on the power analysis for the identity covariance matrix case, and showed that the adaptive test enjoys good overall power. 

In practice, there is this issue of which $q$ to use. Based on the recommendation in \cite{zhangwangshao20}, we set $q=6$. As mentioned in the latter paper, using larger $q$ leads to more trimming and more computational cost. As we demonstrate in the simulations, using $q=6$ and combining with $q=2$ show a very promising performance; see Section 4 for more details.

\section{Ratio-consistent estimator for $\|\Sigma\|_q^q$}
 Notice that the test statistic $T_n(k)$ requires a ratio-consistent estimator for $\|\Sigma\|_q^q$. For example, when $q = 2$, this can be simplified to $\|\Sigma\|_F^2$. The ratio-consistent estimator for $\|\Sigma\|_F^2$ has been proposed in \cite{chenqin10}, but it seems difficult to generalize to $\|\Sigma\|_q^q$. In this section, we introduce a new class of ratio-consistent estimator for $\|\Sigma\|_q^q$ based on U-statistics. We first show the result when $q= 2$ and generalize it to $q\in 2\N$ later. 

	Assume $\{X_t\}_{t = 1}^n \in \mathbb{R}^p$ are i.i.d. random vectors with mean $\bm \mu$ and variance $\Sigma$. Define 
	\begin{equation}
	    \widehat{\|\Sigma\|_F^2} = \frac{1}{4{n \choose 4}}\sum_{1 \leq j_1 < j_2 < j_3 < j_4 \leq n}tr\left((X_{j_1} - X_{j_2})(X_{j_1} - X_{j_2})^T(X_{j_3} - X_{j_4})(X_{j_3} - X_{j_4})^T\right),
	\end{equation}
	as an estimator of $\|\Sigma\|_F^2$.
	\begin{theorem}\label{thm:ratio}
	Under Assumption 1 and Cum(4)$\le C\|\Sigma\|_F^4$ in Assumption 2, $\widehat{\|\Sigma\|_F^2}$ is a ratio-consistent estimator of $\|\Sigma\|_F^2$. i.e. $\widehat{\|\Sigma\|_F^2}/\|\Sigma\|_F^2 \xrightarrow{p} 1$.
	\end{theorem}
	 
	Now we extend this idea to general $q\in 2\N$. We let
	$$\widehat{\|\Sigma\|_q^q} = \frac{1}{2^q{n \choose 2q}}\sum_{l_1,l_2 = 1}^p\sum_{1 \leq i_1 < \cdots < i_q < j_1 < \cdots < j_q \leq n} \prod_{k = 1}^q (X_{i_k,l_1} - X_{j_k,l_1})(X_{i_k,l_2} - X_{j_k,l_2}),$$
	as an estimator for $\|\Sigma\|_q^q$, for any finite positive even number $q$. We can see that the proposed estimator is unbiased through the following proposition.
	
	\begin{proposition} \label{prop:unbias}
	$\widehat{\|\Sigma\|_q^q}$ is an unbiased estimator of $\|\Sigma\|_q^q$.
	\end{proposition}
	
	\begin{proof}[Proof of Proposition \ref{prop:unbias}]
	Since $\{X_t\}_{t = 1}^n$ are i.i.d., 
	\begin{align*}
	    \E[\widehat{\|\Sigma\|_q^q}] &= \frac{1}{2^q{n \choose 2q}}\sum_{l_1,l_2 = 1}^p\sum_{1 \leq i_1 < \cdots < i_q < j_1 < \cdots < j_q \leq n} \prod_{k = 1}^q \E[(X_{i_k,l_1} - X_{j_k,l_1})(X_{i_k,l_2} - X_{j_k,l_2})]\\
	    &= \frac{1}{2^q{n \choose 2q}}\sum_{l_1,l_2 = 1}^p\sum_{1 \leq i_1 < \cdots < i_q < j_1 < \cdots < j_q \leq n} \prod_{k = 1}^q (2\Sigma_{l_1,l_2})\\
	    & = \frac{1}{2^q{n \choose 2q}}\sum_{l_1,l_2 = 1}^p{n \choose 2q}2^q\Sigma_{l_1,l_2}^q = \|\Sigma\|_q^q.
	\end{align*}
	This completes the proof.
	\end{proof}

	The ratio consistency can be shown under the following assumption.
	\begin{assumption} {\label{ass:lq}}
	We assume that 
	\begin{enumerate}
	\item there exists $c > 0$ such that $\inf_{i = 1,...,p}\Sigma_{i,i} > c$;
	    \item there exists $C > 0$ and $r > 2$ such that for $h = 2,3,4$ and $1 \leq l_1 \leq \cdots \leq l_h \leq p$,
	    $$|cum(X_{0,l_1},...,X_{0,l_h})| \leq C(1 \vee (l_h - l_1))^{-r}.$$
	\end{enumerate}
	\end{assumption}
	Notice that Assumption \ref{ass:lq}(2) is required for the ratio consistency, which is weaker than Assumption 4. \change{The  Assumptions 1-5 required for our theory do not state the explicit relationship between $n$ and $p$. For example, when $\Sigma = I_p$, which means there is no cross-sectional dependence, all the assumptions are satisfied and $(n,p)$ can go to infinity freely without any restrictions. When there is cross-sectional dependence, our assumptions may implicitly restrict the relative scale of $n$ and $p$. In general, larger $p$ is a blessing in our setting and it makes the asymptotic approximation more accurate and larger $n$ is always preferred owing to large sample approximation. On the other hand, the computational cost increase when both the dimension and the sample size get large. }
	\begin{theorem} \label{thm:lq}
	Under Assumption~\ref{ass:lq}, $\widehat{\|\Sigma\|_q^q}$ is a ratio-consistent estimator of $\|\Sigma\|_q^{q}$, i.e., $\widehat{\|\Sigma\|_q^q}/\|\Sigma\|_q^q \xrightarrow{p} 1$.
	\end{theorem}
	It is worth noting that implementing the above estimator may be time-consuming for large $q$. In practice, we can always take a random sample of all possible indices and form an incomplete U-statistic to approximate. The consistency of incomplete U-statistic can also be established but not pursued for simplicity. 

\section{Simulation Results}

We compare the monitoring procedures for $q = 2, q = 6$ and $q = (2,6)$ combined. We consider $(n,p) = (100,50)$ with $T = 2$, where the observations $X_i\sim N(\mu_i, \Sigma)$ are generated independently \change{over time}. We consider four possible choices of $\Sigma$, 
\[
    \Sigma_{ij} =  \rho^{|i-j|} \text{ for $\rho =0, 0.2,0.5,0.8$}, 
\]
to evaluate the performance of the monitoring scheme for independent\change{-components} setting or under weak and strong dependence between components. All tests have nominal size $\alpha = 0.1$. 

Under the null $H_0$, there is no change point, $\mu_i = 0$ for all $i$. For the alternative, we consider $\mu_i = \sqrt{\delta/r}(\bm 1_r, \bm 0_{p-r})$ for $i = (\floor{1.25n}+1), \ldots, nT$. Under the dense alternative, we set $(\delta, r) = (1,p), (2,p)$.   Under the sparse alternative, we set $(\delta,r) = (1,3),(1,1)$.

To illustrate the finite sample performance of our monitoring statistics, we compare with \cite{mei:10} (denoted as Mei) and \cite{liu:19} (denoted as LZM), which are similar to the open-end scenario in 
\cite{chu:96}. Both methods do not require Phase I data and are originally designed to minimize the average run length. Therefore, they do not explicitly control the Type-I error. To make a fair comparison with the current methods, which are proposed under the closed-end monitoring framework, we generate $n$ independent Gaussian sample from $N(\bm 0, \bm I_{p\times p})$ and calculate Mei and LZM monitoring statistics. We empirically determine the critical value such that the empirical rejection rate is $10\%$ based on $2500$ simulated datasets. For Mei's methods, we need to specify the distribution after the change point, which we set it to be the distribution under the alternative $(\delta, r) = (1,p)$. For LZM's method, we use the same setting in \cite{liu:19} and set $b = \log(10), \rho = 0.25, t = 4$ and $s = 1$.

Table \ref{tab:size} shows the size of the monitoring procedure for the benchmark methods and the proposed methods for three different boundary functions T1, T2, T3 introduced in Section 2.1  under different correlation coefficients $\rho$. Notice that the size is noticeably worse for $\rho = 0.8$. This is partially due to the poor performance in the ratio-consistent estimator since its variance increases as the \change{cross-sectional} dependence increases. Also, please note that the size seems to go in different directions for $q = 2$ and $q=6$ as the correlation increases. The combined test, on the other hand, balances out such distortions. To make sure this is only a finite sample behavior, we increase $(n,p)$ from $(100,50)$ to $(200,200)$, the size distortion for all tests improved noticeably for almost all settings. The additional results are available in the Supplementary Materials. By contrast, Mei and LZM only achieved correct size for the independent\change{-component} case, since we select the threshold under the exact same setting. However, when there is \change{cross-sectional} dependence between different data streams, the size is no longer controlled and the size distortion is much more severe than the $L_q$ based tests. 

\begin{table}[!h]
\centering
\caption{Size of different monitoring procedures}
\label{tab:size}
\setlength{\tabcolsep}{1.5pt}
\begin{tabular}{lrrrrrrrrrrr}
\hline
 & \multicolumn{1}{c}{}    & \multicolumn{1}{c}{}    & \multicolumn{3}{c}{T1}                                                     & \multicolumn{3}{c}{T2}                                                     & \multicolumn{3}{c}{T3}                                                     \\ \cline{4-12} 
$\alpha = 0.1$     & \multicolumn{1}{c}{Mei} & \multicolumn{1}{c}{LZM} & \multicolumn{1}{c}{$L_2$} & \multicolumn{1}{c}{$L_6$} & \multicolumn{1}{c}{Comb} & \multicolumn{1}{c}{$L_2$} & \multicolumn{1}{c}{$L_6$} & \multicolumn{1}{c}{Comb} & \multicolumn{1}{c}{$L_2$} & \multicolumn{1}{c}{$L_6$} & \multicolumn{1}{c}{Comb} \\ \hline
$\rho=0$           & 0.094                   & 0.105                   & 0.086                  & 0.048                  & 0.067                    & 0.093                  & 0.045                  & 0.071                    & 0.097                  & 0.045                  & 0.070                    \\
$\rho=0.2$         & 0.058                   & 0.125                   & 0.083                  & 0.048                  & 0.057                    & 0.082                  & 0.045                  & 0.055                    & 0.083                  & 0.046                  & 0.051                    \\
$\rho=0.5$         & 0.002                   & 0.176                   & 0.103                  & 0.048                  & 0.084                    & 0.104                  & 0.048                  & 0.082                    & 0.108                  & 0.048                  & 0.080                    \\
$\rho=0.8$         & 0.000                   & 0.409                   & 0.135                  & 0.028                  & 0.085                    & 0.145                  & 0.027                  & 0.093                    & 0.137                  & 0.026                  & 0.086                    \\ \hline
\end{tabular}
\end{table}

Table \ref{tab:power_dense} provides the power result (left column) and ADT (right column) for different tests under dense alternatives. As expected, the $L_2$-based test demonstrates higher power compared to that of the $L_6$-based test. The power of the combined test falls in between and is closer to the power of $L_2$-based test. As the correlation increases, the powers of all  tests decrease due to reduced signal. Among three different boundary functions, T2 seems to be the one with the shortest average delay time (ADT) with a slight sacrifice in power. Mei's method is only better than the $L_6$ based test when there is no strong \change{cross-sectional} dependence, and is generally worse than all other methods and have relatively longer delay even when the distribution under the alternative is correctly specified. Notice that when $\rho = 0.8$, Mei's method lost the power completely. LZM in general has the slightly shorter detection delay but at the cost of a much lower power compared with $L_2$ based test and the combined test. This means the LZM is quicker in signaling an alarm when a change point is detected. Although LZM showed good power for the strong \change{cross-sectional} dependence case compared with the combined test, it comes at the price of much distorted size. This is related to the fact that LZM assumes all data streams are independent.
\begin{table}[]
\centering
\caption{Power under dense alternatives}
\label{tab:power_dense}
\setlength{\tabcolsep}{1.5pt}
\renewcommand{\arraystretch}{0.7}
\begin{tabular}{cc|rr|rr|c|rr|rr|rr}
\hline
Power                         &                          & \multicolumn{2}{c|}{Mei}                             & \multicolumn{2}{c|}{LZM}                             &        & \multicolumn{2}{c|}{$L_2$}                           & \multicolumn{2}{c|}{$L_6$}                           & \multicolumn{2}{c}{Comb}                           \\ \cline{3-13} 
$\alpha = 0.1$                & $(\delta,r)$             & \multicolumn{1}{c}{power} & \multicolumn{1}{c|}{ADT} & \multicolumn{1}{c}{power} & \multicolumn{1}{c|}{ADT} & $w(t)$ & \multicolumn{1}{c}{power} & \multicolumn{1}{c|}{ADT} & \multicolumn{1}{c}{power} & \multicolumn{1}{c|}{ADT} & \multicolumn{1}{c}{power} & \multicolumn{1}{c}{ADT} \\ \hline
\multirow{6}{*}{$\rho = 0.0$} & \multirow{3}{*}{$(1,p)$} & \multirow{3}{*}{0.852}    & \multirow{3}{*}{72.9}    & \multirow{3}{*}{0.628}    & \multirow{3}{*}{38.0}    & T1     & 0.958                     & 51.9                     & 0.295                     & 64.6                     & 0.926                     & 55.0                    \\
                              &                          &                           &                          &                           &                          & T2     & 0.951                     & 44.3                     & 0.284                     & 63.0                     & 0.921                     & 47.7                    \\
                              &                          &                           &                          &                           &                          & T3     & 0.953                     & 46.8                     & 0.286                     & 63.4                     & 0.921                     & 50.2                    \\ \cline{2-13} 
                              & \multirow{3}{*}{$(2,p)$} & \multirow{3}{*}{0.999}    & \multirow{3}{*}{69.3}    & \multirow{3}{*}{1.000}    & \multirow{3}{*}{15.1}    & T1     & 1.000                     & 27.5                     & 0.919                     & 56.2                     & 1.000                     & 29.5                    \\
                              &                          &                           &                          &                           &                          & T2     & 1.000                     & 20.4                     & 0.919                     & 54.3                     & 1.000                     & 21.9                    \\
                              &                          &                           &                          &                           &                          & T3     & 1.000                     & 22.9                     & 0.920                     & 54.9                     & 1.000                     & 24.7                    \\ \hline
\multirow{6}{*}{$\rho = 0.2$} & \multirow{3}{*}{$(1,p)$} & \multirow{3}{*}{0.740}    & \multirow{3}{*}{73.3}    & \multirow{3}{*}{0.675}    & \multirow{3}{*}{38.2}    & T1     & 0.935                     & 51.8                     & 0.302                     & 64.4                     & 0.907                     & 54.9                    \\
                              &                          &                           &                          &                           &                          & T2     & 0.930                     & 44.2                     & 0.291                     & 62.9                     & 0.906                     & 47.7                    \\
                              &                          &                           &                          &                           &                          & T3     & 0.933                     & 46.7                     & 0.294                     & 63.5                     & 0.903                     & 50.3                    \\ \cline{2-13} 
                              & \multirow{3}{*}{$(2,p)$} & \multirow{3}{*}{1.000}    & \multirow{3}{*}{69.9}    & \multirow{3}{*}{1.000}    & \multirow{3}{*}{15.6}    & T1     & 1.000                     & 28.0                     & 0.884                     & 56.6                     & 1.000                     & 30.0                    \\
                              &                          &                           &                          &                           &                          & T2     & 1.000                     & 20.8                     & 0.884                     & 54.8                     & 1.000                     & 22.3                    \\
                              &                          &                           &                          &                           &                          & T3     & 1.000                     & 23.4                     & 0.883                     & 55.3                     & 1.000                     & 25.2                    \\ \hline
\multirow{6}{*}{$\rho = 0.5$} & \multirow{3}{*}{$(1,p)$} & \multirow{3}{*}{0.243}    & \multirow{3}{*}{74.1}    & \multirow{3}{*}{0.715}    & \multirow{3}{*}{34.3}    & T1     & 0.844                     & 52.9                     & 0.274                     & 63.3                     & 0.796                     & 55.8                    \\
                              &                          &                           &                          &                           &                          & T2     & 0.843                     & 45.2                     & 0.267                     & 61.5                     & 0.787                     & 47.9                    \\
                              &                          &                           &                          &                           &                          & T3     & 0.847                     & 47.9                     & 0.267                     & 62.0                     & 0.792                     & 50.7                    \\ \cline{2-13} 
                              & \multirow{3}{*}{$(2,p)$} & \multirow{3}{*}{0.932}    & \multirow{3}{*}{72.2}    & \multirow{3}{*}{1.000}    & \multirow{3}{*}{15.7}    & T1     & 1.000                     & 30.7                     & 0.864                     & 55.9                     & 1.000                     & 33.0                    \\
                              &                          &                           &                          &                           &                          & T2     & 1.000                     & 23.1                     & 0.861                     & 54.2                     & 1.000                     & 24.8                    \\
                              &                          &                           &                          &                           &                          & T3     & 1.000                     & 25.7                     & 0.861                     & 54.8                     & 1.000                     & 27.8                    \\ \hline
\multirow{6}{*}{$\rho = 0.8$} & \multirow{3}{*}{$(1,p)$} & \multirow{3}{*}{0.000}    & \multirow{3}{*}{NA}      & \multirow{3}{*}{0.803}    & \multirow{3}{*}{29.0}    & T1     & 0.632                     & 54.6                     & 0.165                     & 62.5                     & 0.560                     & 56.8                    \\
                              &                          &                           &                          &                           &                          & T2     & 0.637                     & 46.4                     & 0.162                     & 60.9                     & 0.575                     & 48.6                    \\
                              &                          &                           &                          &                           &                          & T3     & 0.642                     & 49.4                     & 0.162                     & 61.4                     & 0.568                     & 51.8                    \\ \cline{2-13} 
                              & \multirow{3}{*}{$(2,p)$} & \multirow{3}{*}{0.001}    & \multirow{3}{*}{74.0}    & \multirow{3}{*}{0.997}    & \multirow{3}{*}{16.1}    & T1     & 0.990                     & 38.3                     & 0.666                     & 56.0                     & 0.984                     & 40.8                    \\
                              &                          &                           &                          &                           &                          & T2     & 0.990                     & 30.1                     & 0.663                     & 54.2                     & 0.983                     & 32.1                    \\
                              &                          &                           &                          &                           &                          & T3     & 0.990                     & 32.7                     & 0.666                     & 54.9                     & 0.983                     & 35.4                    \\ \hline
\end{tabular}
\end{table}

Table \ref{tab:power_sparse} provides power of different tests under sparse alternatives.  The $L_6$-based test and the combined test are very comparable in power and $L_2$-based test exhibits inferior power in most settings as expected.  One interesting observation is that for the case $(\delta, r) = (1,3)$, the $L_2$-based test still shows slightly higher power than the $L_6$-based test when $\rho = 0.2$, which means that for this particular setting, the change is not ``sparse'' enough. As the correlation increases, we observe a noticeable drop in power, which is  similar to the dense alternative setting and is again attributed to the reduced signal size. Among three different boundary functions, T2 still has shortest average delay time with a slight power loss compared to other two boundary functions. Mei's method has worse power since it is designed for the dense signals and the distribution under the alternative is misspecified. By comparison, LZM gives consistently good power and short delay time across all settings. However, the good power under strong  \change{cross-sectional} dependence is still offset by the severe size distortion under the null.

Apart from evaluating the size and power of the monitoring procedure, we also compare the computational cost of the recursive formulation versus the brute force approach. For the case of $(n,p)=(100,50)$, the average run-time of the brute force approach is 12.89 times of the recursive algorithm  under $H_0$,  and  the average run-time of the  brute force approach is 13.07 times of that for the recursive algorithm under the alternative. The codes are implemented in R. This demonstrates the substantial efficiency gain from the recursive  computational algorithm. 
\begin{table}[]
\caption{Power under sparse alternatives}
\label{tab:power_sparse}
\setlength{\tabcolsep}{1.5pt}
\renewcommand{\arraystretch}{0.7}
\begin{tabular}{cc|rr|rr|c|rr|rr|rr}
\hline
Power                         &                          & \multicolumn{2}{c|}{Mei}                             & \multicolumn{2}{c|}{LZM}                             &        & \multicolumn{2}{c|}{$L_2$}                           & \multicolumn{2}{c|}{$L_6$}                           & \multicolumn{2}{c}{Comb}                            \\ \cline{3-13} 
$\alpha = 0.1$                & $(\delta,r)$             & \multicolumn{1}{c}{power} & \multicolumn{1}{c|}{ADT} & \multicolumn{1}{c}{power} & \multicolumn{1}{c|}{ADT} & $w(t)$ & \multicolumn{1}{c}{power} & \multicolumn{1}{c|}{ADT} & \multicolumn{1}{c}{power} & \multicolumn{1}{c|}{ADT} & \multicolumn{1}{c}{power} & \multicolumn{1}{c}{ADT} \\ \hline
\multirow{6}{*}{$\rho = 0.0$} & \multirow{3}{*}{$(1,3)$} & \multirow{3}{*}{0.422}    & \multirow{3}{*}{74.0}    & \multirow{3}{*}{0.990}    & \multirow{3}{*}{27.4}    & T1     & 0.976                     & 51.5                     & 0.999                     & 37.8                     & 0.999                     & 40.5                    \\
                              &                          &                           &                          &                           &                          & T2     & 0.967                     & 43.8                     & 0.999                     & 35.9                     & 0.999                     & 38.0                    \\
                              &                          &                           &                          &                           &                          & T3     & 0.972                     & 46.4                     & 0.999                     & 36.6                     & 0.999                     & 39.0                    \\ \cline{2-13} 
                              & \multirow{3}{*}{$(1,1)$} & \multirow{3}{*}{0.400}    & \multirow{3}{*}{73.9}    & \multirow{3}{*}{1.000}    & \multirow{3}{*}{23.4}    & T1     & 0.961                     & 51.5                     & 0.951                     & 51.0                     & 0.976                     & 52.2                    \\
                              &                          &                           &                          &                           &                          & T2     & 0.958                     & 44.1                     & 0.953                     & 49.5                     & 0.974                     & 46.3                    \\
                              &                          &                           &                          &                           &                          & T3     & 0.959                     & 46.4                     & 0.953                     & 50.0                     & 0.976                     & 48.7                    \\ \hline
\multirow{6}{*}{$\rho = 0.2$} & \multirow{3}{*}{$(1,3)$} & \multirow{3}{*}{0.274}    & \multirow{3}{*}{74.1}    & \multirow{3}{*}{0.990}    & \multirow{3}{*}{29.1}    & T1     & 0.946                     & 52.2                     & 0.937                     & 51.6                     & 0.955                     & 52.6                    \\
                              &                          &                           &                          &                           &                          & T2     & 0.939                     & 44.6                     & 0.935                     & 50.0                     & 0.955                     & 47.1                    \\
                              &                          &                           &                          &                           &                          & T3     & 0.943                     & 47.1                     & 0.936                     & 50.5                     & 0.954                     & 49.2                    \\ \cline{2-13} 
                              & \multirow{3}{*}{$(1,1)$} & \multirow{3}{*}{0.268}    & \multirow{3}{*}{74.1}    & \multirow{3}{*}{1.000}    & \multirow{3}{*}{23.9}    & T1     & 0.961                     & 52.6                     & 0.998                     & 37.3                     & 0.999                     & 40.2                    \\
                              &                          &                           &                          &                           &                          & T2     & 0.951                     & 45.3                     & 0.998                     & 35.4                     & 0.999                     & 37.7                    \\
                              &                          &                           &                          &                           &                          & T3     & 0.957                     & 47.6                     & 0.998                     & 36.0                     & 0.999                     & 38.6                    \\ \hline
\multirow{6}{*}{$\rho = 0.5$} & \multirow{3}{*}{$(1,3)$} & \multirow{3}{*}{0.048}    & \multirow{3}{*}{74.5}    & \multirow{3}{*}{0.972}    & \multirow{3}{*}{28.2}    & T1     & 0.871                     & 54.7                     & 0.881                     & 51.5                     & 0.887                     & 53.4                    \\
                              &                          &                           &                          &                           &                          & T2     & 0.856                     & 47.1                     & 0.878                     & 49.8                     & 0.884                     & 48.7                    \\
                              &                          &                           &                          &                           &                          & T3     & 0.860                     & 49.9                     & 0.880                     & 50.4                     & 0.886                     & 50.6                    \\ \cline{2-13} 
                              & \multirow{3}{*}{$(1,1)$} & \multirow{3}{*}{0.036}    & \multirow{3}{*}{74.3}    & \multirow{3}{*}{1.000}    & \multirow{3}{*}{23.2}    & T1     & 0.880                     & 55.9                     & 0.997                     & 38.0                     & 0.997                     & 40.7                    \\
                              &                          &                           &                          &                           &                          & T2     & 0.871                     & 49.1                     & 0.997                     & 36.1                     & 0.997                     & 38.2                    \\
                              &                          &                           &                          &                           &                          & T3     & 0.879                     & 51.2                     & 0.997                     & 36.8                     & 0.997                     & 39.2                    \\ \hline
\multirow{6}{*}{$\rho = 0.8$} & \multirow{3}{*}{$(1,3)$} & \multirow{3}{*}{0.000}    & \multirow{3}{*}{NA}      & \multirow{3}{*}{0.971}    & \multirow{3}{*}{24.7}    & T1     & 0.621                     & 58.9                     & 0.800                     & 52.9                     & 0.808                     & 55.3                    \\
                              &                          &                           &                          &                           &                          & T2     & 0.610                     & 50.6                     & 0.801                     & 51.3                     & 0.802                     & 51.5                    \\
                              &                          &                           &                          &                           &                          & T3     & 0.614                     & 53.7                     & 0.803                     & 51.9                     & 0.807                     & 53.2                    \\ \cline{2-13} 
                              & \multirow{3}{*}{$(1,1)$} & \multirow{3}{*}{0.000}    & \multirow{3}{*}{NA}      & \multirow{3}{*}{1.000}    & \multirow{3}{*}{21.5}    & T1     & 0.602                     & 61.1                     & 0.998                     & 38.8                     & 0.997                     & 41.7                    \\
                              &                          &                           &                          &                           &                          & T2     & 0.588                     & 53.6                     & 0.998                     & 36.8                     & 0.997                     & 39.3                    \\
                              &                          &                           &                          &                           &                          & T3     & 0.601                     & 56.8                     & 0.998                     & 37.5                     & 0.997                     & 40.2                    \\ \hline
\end{tabular}
\end{table}
\section{Data Illustration}

\subsection{Tonnage dataset}

We first propose to apply the proposed methodology to monitor the multi-channel tonnage profile collected in a forging process in \citep{lei2010automatic}, where four different strain gauge sensors are mounted at each column of the forging machine, measuring the exerted force of the press. The setup of the process is shown in Figure \ref{Fig: rollingsetup}. The four strain gauge sensors represent the signature of the product and are used for process monitoring and change detection in \cite{lei2010automatic}. For example, 10 examples of the signals before the changes and after the changes are shown in Figure \ref{Fig: rollingexp}. As mentioned by \cite{lei2010automatic,yan2018real}, a missing part only affects a small region of the signals, which makes it very hard to detect, as shown in Figure \ref{Fig: rollingexp}. 

\begin{figure}
\centering
\subfloat[The setup of the forging machine]{\includegraphics[width=0.46\linewidth,height=4cm]{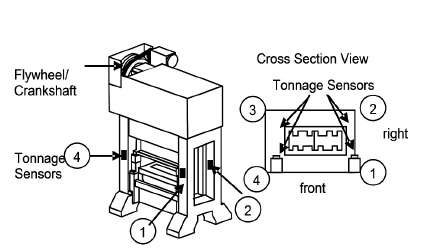} \label{Fig: rollingsetup}}\quad\subfloat[Data collected in the forging machine]{\includegraphics[width=0.46\linewidth,height=4cm]{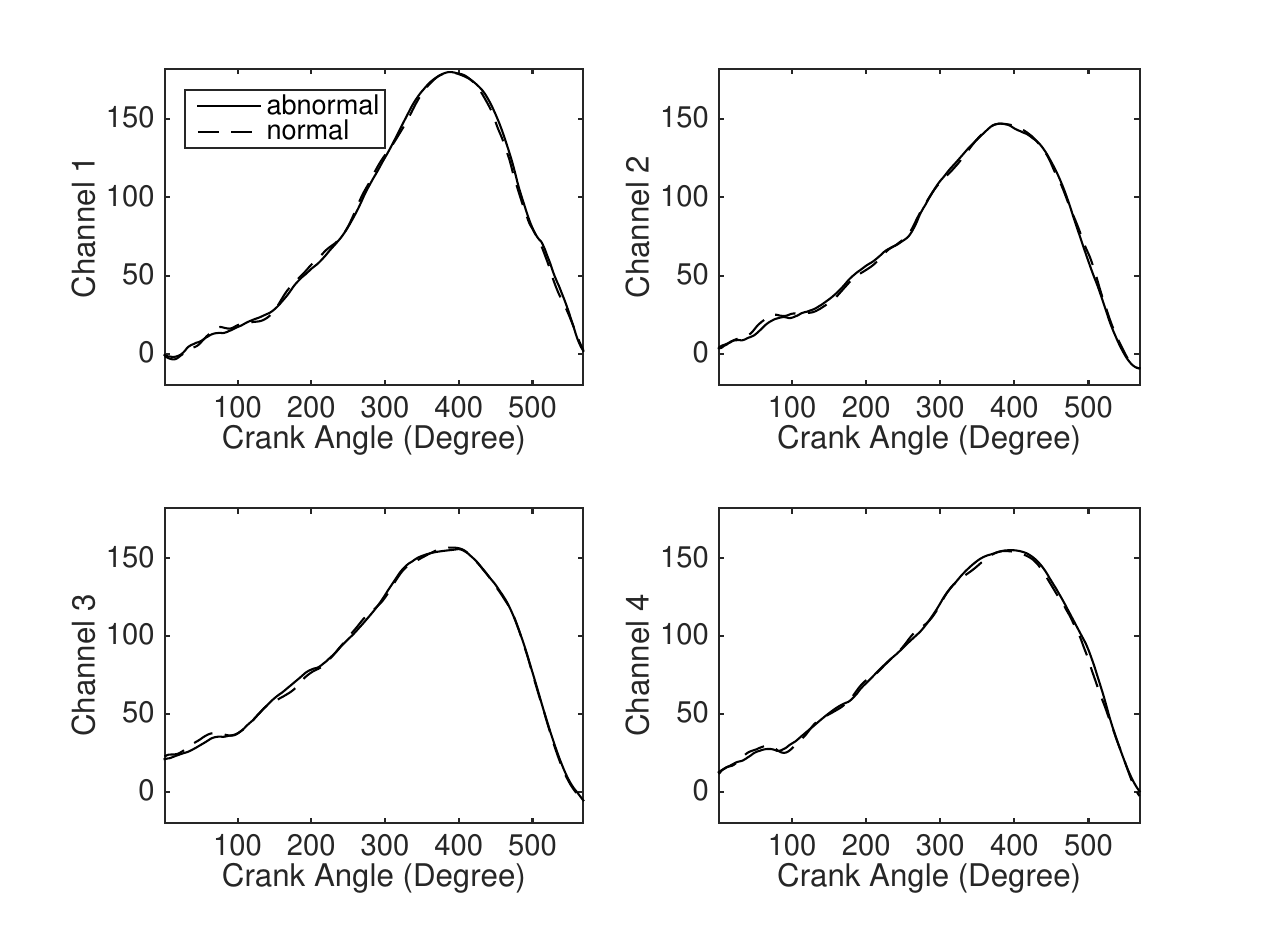} \label{Fig: rollingexp}}
\caption{Forging machine setup and the collected tonnage dataset}
\end{figure}

We select a subset of the data with $n = 200$, where the first 130 observations are from the normal tonnage sample, and the samples after 130 are abnormal. We project the data to $20$-dimensional space by training an anomaly basis on a holdout sample as has been done in \cite{yan2018real}. The first 100 observations are treated as a Phase I stage without any changes and we learn the 2-norm and $q$-norm of the covariance matrix from them. The monitoring scheme started at observation 107 (trimming due to $q= 6$). The $L_6$-based test stopped at time 137, and estimated the possible change point location at time 128 by performing a retrospective test at time 137. The $L_2$ based test signaled slightly earlier at time 135 and also estimated the change at 128. The combined test signaled an alarm at time 135 with the same estimated location. The trajectory of the $L_2$ and the $L_6$ test statistics  are shown in Figure \ref{fig:tonnageL2} and \ref{fig:tonnageL6}, respectively. Notice that, when we set the size of individual test to be $\alpha^* = (1-0.1)^{1/2} = 5.13\%$, the size of the combined test will be $\alpha = 1-\alpha^{*2} = 0.1$. We signal an alarm when at least one test statistic exceeds the corresponding threshold.

\begin{figure}
\centering
\subfloat[$L_2$ based test for tonnage data]{\includegraphics[width=0.46\linewidth]{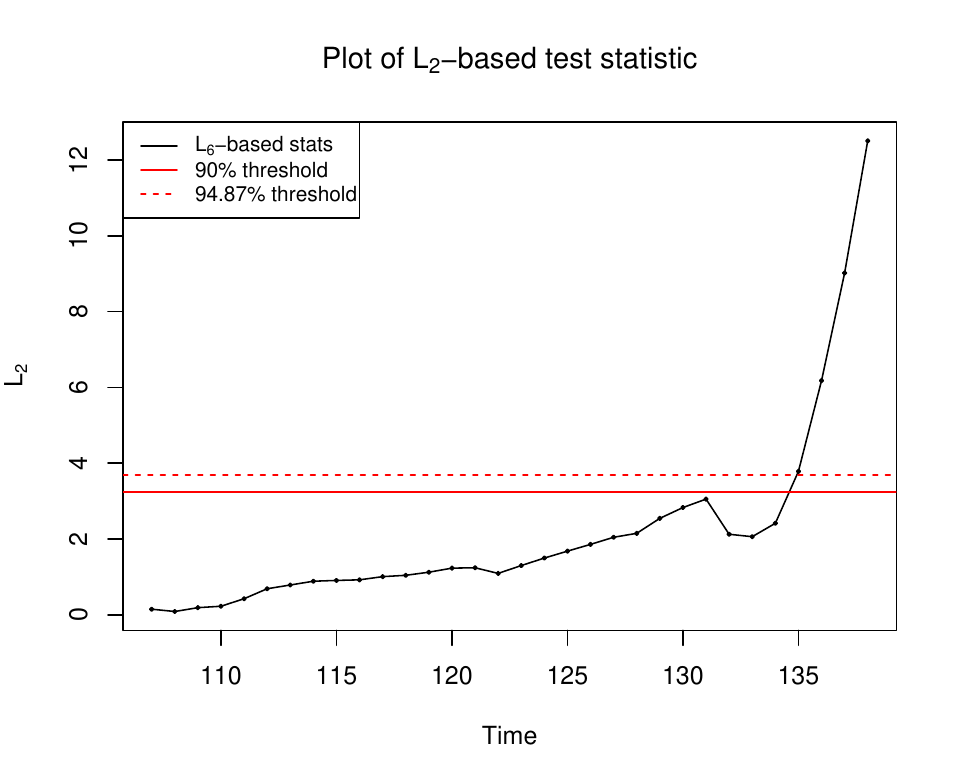} \label{fig:tonnageL2}}\quad\subfloat[$L_6$ based test for tonnage data]{\includegraphics[width=0.46\linewidth]{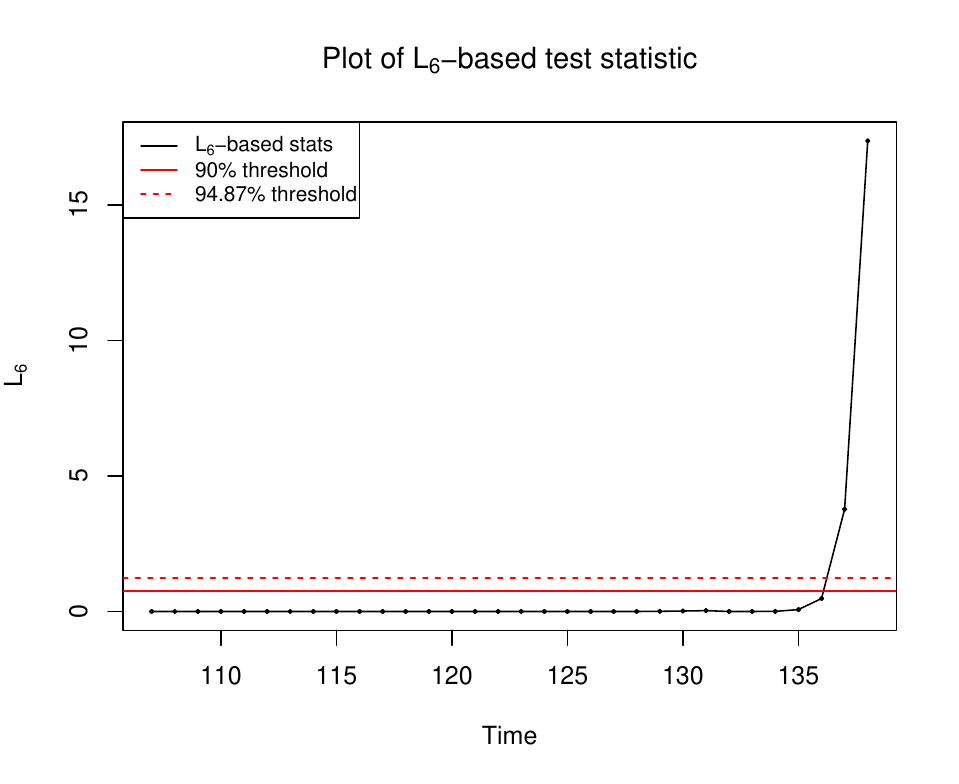} \label{fig:tonnageL6}}
\caption{Testing Statistics for tonnage data}
\end{figure}

\subsection{Rolling dataset}

We then consider the process monitoring in a steel rolling manufacturing process. Surface defects such as seam defects can result in stress concentration on the bulk and may cause failures if the steel bar is used in the product. However, the rolling process is a high-speed process with the rolling bar moving around 200 miles per hour, providing real-time online anomaly detection for the high-speed rolling process is very important to prevent further the product damage. 

 The dataset is collected in such high-speed rolling process. Here, we have selected a segment near the end of the rolling bar, which  contains 100 images of the rolling process. To remove the trend, we have applied a smooth background remover and further downsample the image to only $16\times 64$ pixels. An example of the normal image and the abnormal image are shown in Figure \ref{fig:rolling1} and \ref{fig:rolling91}, respectively. 

\begin{figure}
\centering
\subfloat[Normal rolling image]{\includegraphics[width=0.45\linewidth]{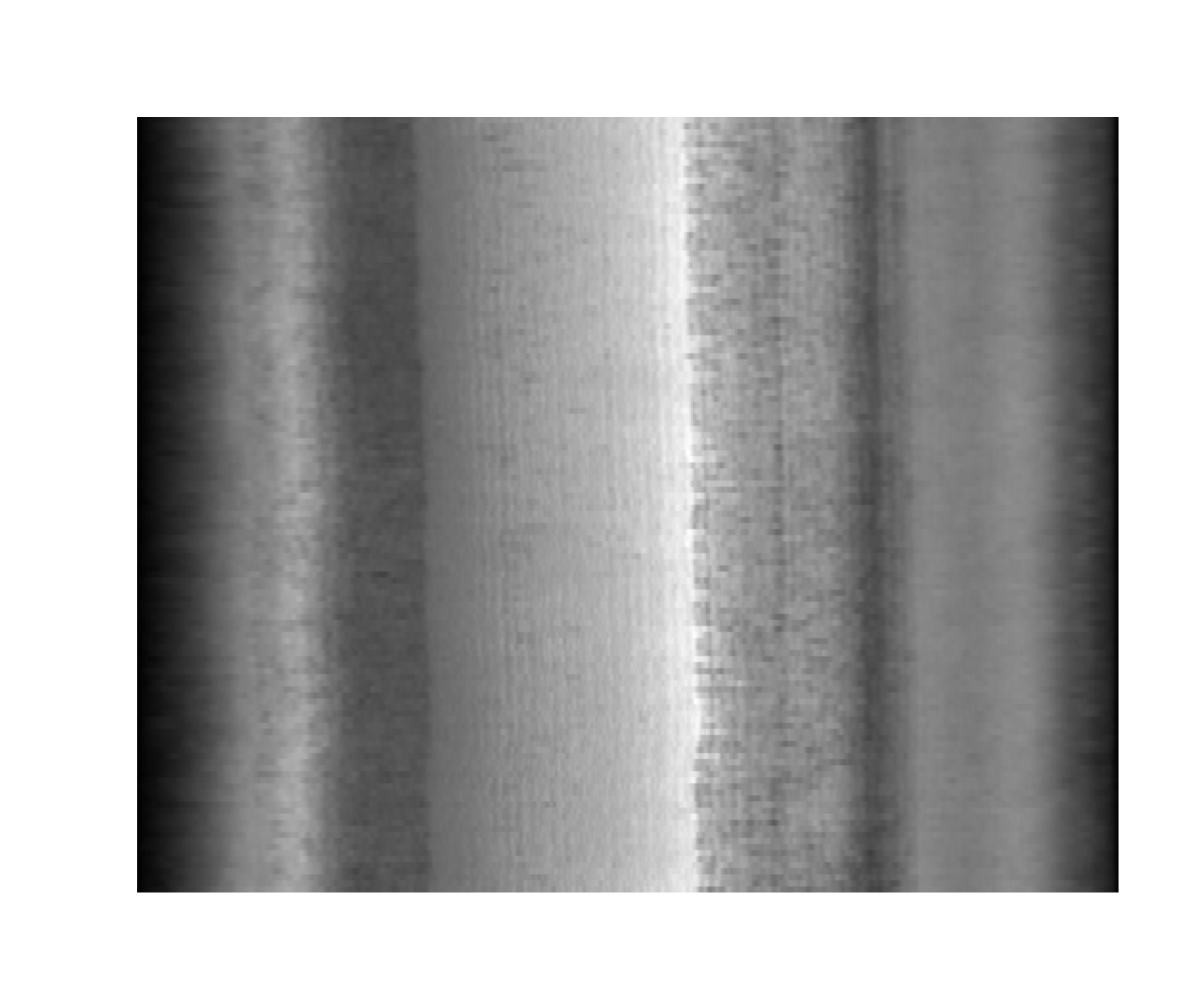} \label{fig:rolling1}}\quad\subfloat[Abnormal rolling image]{\includegraphics[width=0.46\linewidth]{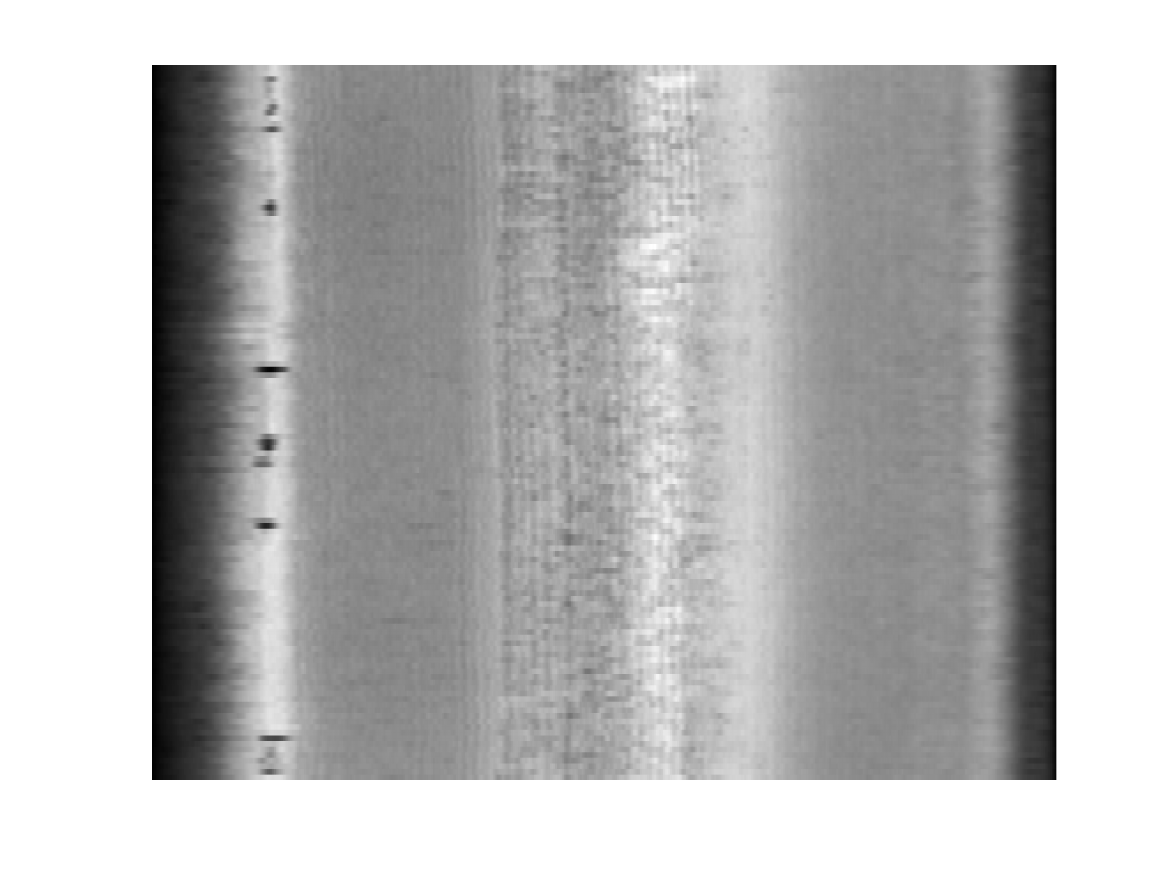} \label{fig:rolling91}}
\caption{Examples of the rolling images}
\end{figure}
  
We treated the first 50 observations as training set and obtained ratio-consistent estimators  $\widehat{\|\Sigma\|_q^q}$. After performing the change point monitoring procedure, the $L_6$-norm-based test signals an alarm at the time $97$ and estimated that the possible change point location is at time $89$ based on the retrospective test. On the other hand, the $L_2$ based test fails to detect the change within the finite time horizon. The combined test also signals the alarm at time 97. We present the rolling image at time 91 in Figure \ref{fig:rolling91}. This shows that after downsampling, the change is still quite sparse. The adaptive monitoring procedure is still powerful as long as one test has power. We also present the trajectory of the test statistic at each time point in Figure \ref{fig:rollingL2} and \ref{fig:rollingL6}. Notice that there is a downshift in the $L_2$-based monitoring statistic right after the estimated change. This is due to the fact that the signal is very sparse, and the construction of our proposed statistic may  admit negative values for a short period of time. The negative values here should not be a major concern as the test statistic should admit positive values in probability under the alternatives. We confirmed this by adding an artificial dense change in the data. On the other hand, $L_6$-based monitoring statistics detect the change efficiently due to the ability to capture the sparse change in the system.

\begin{figure}
\centering
\subfloat[$L_2$ based test for rolling data]{\includegraphics[width=0.46\linewidth, height=5.5cm]{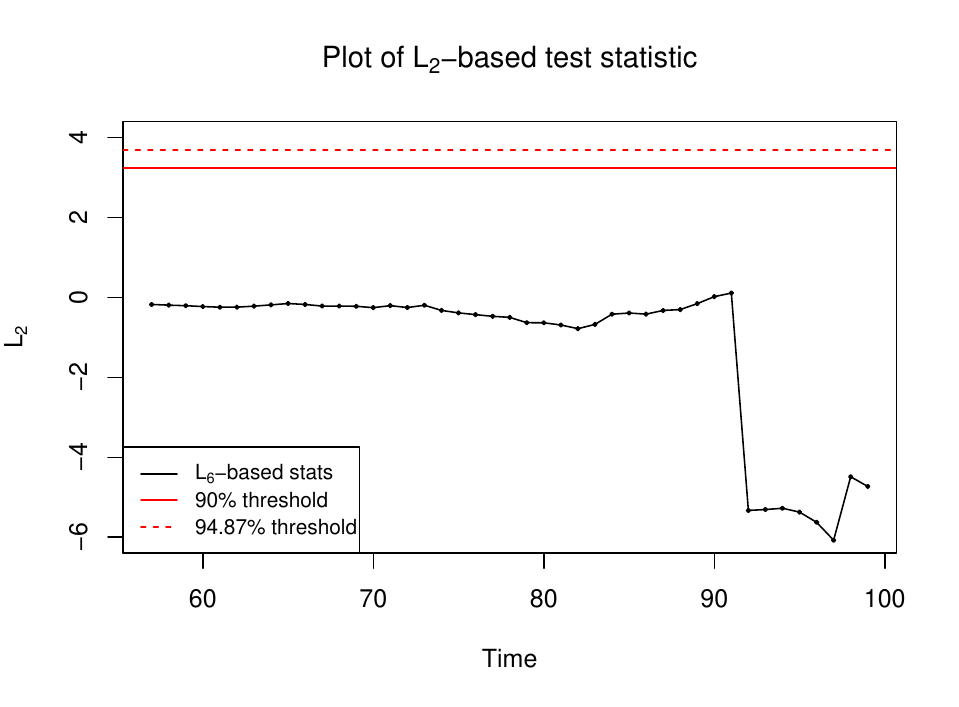} \label{fig:rollingL2}}\quad\subfloat[$L_6$ based test for rolling data]{\includegraphics[width=0.46\linewidth, height=5.5cm]{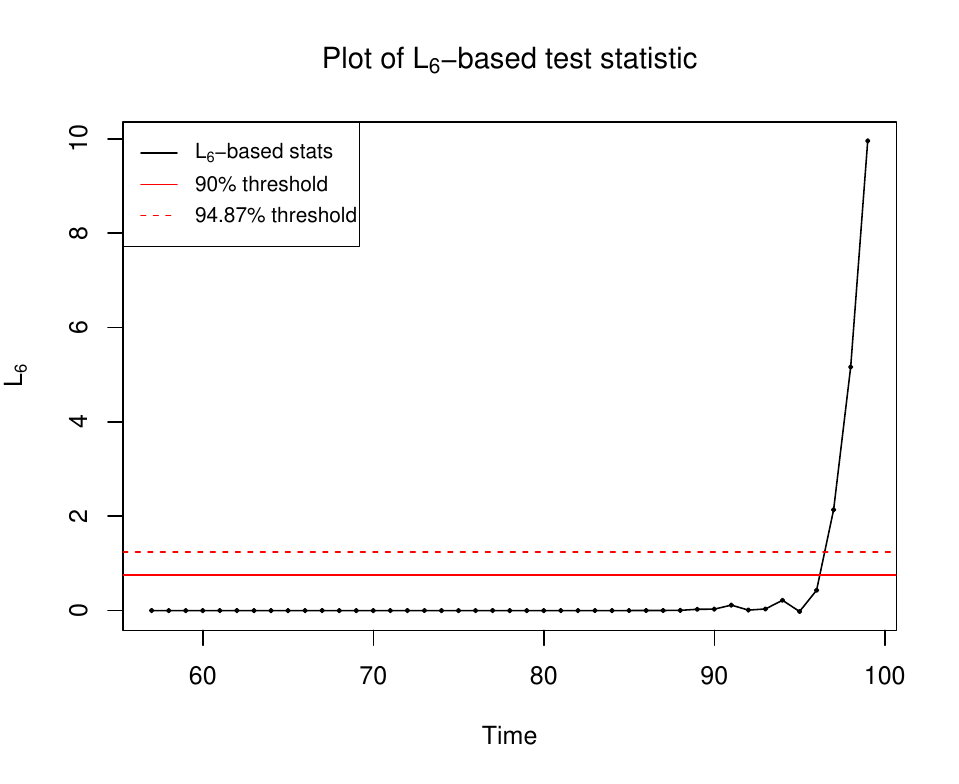} \label{fig:rollingL6}}
\caption{Examples of the rolling images}
\end{figure}

\section{Summary and Conclusion}

In this article, we propose a new methodology to monitor a mean shift in \change{ temporally independent} high-dimensional observations. 
Our change point monitoring method  targets at the $L_q$-norm of the mean change for $q=2,4,6,\cdots$. By combining the monitoring statistics for different values of $q\in 2\N$, the adaptive procedure achieves overall satisfactory power against both sparse and dense changes in the mean. This can be very desirable from a practitioner's viewpoint as often we do not have the knowledge about the type of alternatives. Compared with the recently developed methods for monitoring large-scale data streams [e.g., \cite{mei:10}, \cite{xie2013sequential}, \cite{liu:19}], our method is fully nonparametric and does not require strong distributional assumptions. Furthermore, our method allows for certain  \change{cross-sectional} dependence across data streams, which could naturally arise in many applications. 

To conclude, we mention a few interesting future directions. Firstly, our focus in this paper is on the mean change, and it is natural to ask whether the method can be extended to monitor a change in the covariance matrix. Secondly, many streaming data have weak dependence over time due to its sequential nature, and how to accommodate weak temporal dependence will be of interest. Thirdly, in the current implementation, the ratio-consistent estimators are learned from the training data and do not change as more observations are available. In practice, if the monitoring scheme runs for a long time without signaling an alarm, it might be helpful to periodically update ratio-consistent estimators to gain efficiency, especially when the initial training sample is short. \change{However, it may be impractical to
update this estimator for each $k$ since there seems no easy recursive way to update
this estimator and the associated computational cost is high. The user might need
to determine how often to update it based on the actual computational resources.} Fourthly, even though the proposed algorithm can detect the sparse change, in many applications, it is also an important problem to identify which individual data stream actually experiences a change, which will be left for future research. 


\section*{Supplementary Materials}
The supplementary materials contains technical proofs for the theoretical results as well as additional simulation results.
\begin{proof}[Proof of Theorem 1]
We can directly apply the results shown in \cite{wang2019inference} for the partial sum process 
\[S_n(a,b) = \sum_{i = \floor{na} + 1}^{\floor{nb}-1}\sum_{j = \floor{na} + 1}^{i} X_{i+1}^T X_j.\] The partial sum process 
\[\Big\{ \frac{\sqrt{2}}{n||\Sigma||_F} S_n(a,b)\Big\}_{(a,b) \in [0,T]^2} \rightsquigarrow Q \quad in \quad l^\infty([0,T]^2)\]
where $Q$ is a Gaussian process whose covariance structure is the following
\begin{equation*}
    Cov(Q(a_1,b_1),(a_2,b_2)) = \begin{cases}
    (\min(b_1,b_2) - \max(a_1,a_2))^2 & if\quad \max(a_1,a_2) \le \min(b_1,b_2)\\
    0 & otherwise
    \end{cases}
\end{equation*}
The test statistic is a continuous transformation of the Gaussian process and the results stated follows. 
\end{proof}
\begin{proof}[Proof of Theorem 2]
We now analyze the power of the first proposed test. Suppose the change point is at $k^*$, where $k^*/n \to r$ for some constant $r \in (1, T)$. This assures that the change point does not occur extremely early or late in the monitoring period. Under the alternative hypothesis, define a new sequence of random vectors $Y_i$,
\[Y_ i = \begin{cases} X_i & i = 1,\ldots,k^*\\
X_i - \Delta & i = k^*+1,\ldots, n
\end{cases}. \]
This sequence does not have a change point. Without loss of generosity, assume $Y_i$'s are centered.

Suppose that 
\[\frac{n\Delta^T\Delta}{||\Sigma||_F} \to b \in [0 +\infty).\]
When $m < k < k^*$, $G_k(m)$ statistic will not be affected. It suffices to consider the case $m < k^* <k$ and $k^*< m < k$.
Following the decomposition in \cite{wang2019inference}, 
under the fixed alternative when $k^* > m$,
\begin{align*}
     G_k(m) & = G^Y_k(m) + {(k-k^*)(k-k^*-1)m(m-1)}||\Delta||_2^2 \\
     & \quad -  {2(k-k^*)(k-m-2)(m-1)}\sum_{j = 1}^m Y_j^T \Delta \\
     & \quad - 4(m-1)(m-2)(k-k^*)\sum_{j = m+1}^{k^*}Y_j^T \Delta.
\end{align*}
$G_n^Y(m)$ is the statistic calculated for the $Y_i$ sequence. Let $s_n(k) = \sum_{j = 1}^k Y_j^T \Delta$. Then 
\[\sup_{1 \le l\le k \le  nT}|\sum_{j = l}^k Y_j^T \Delta | \le 2 \sup_{1\le k \le nT}|s_n(k)|  = O_p(n^{1/2} (\Delta^T \Sigma \Delta)^{1/2}). \]
The last part is obtained by Kolmogorov's inequality. This implies that when $k^* > m$,
\begin{align*}
    \frac{1}{n^3\|\Sigma\|_F}G_k(m)
    & =\frac{1}{n^3\|\Sigma\|_F} G_k^Y(m) +  \frac{(k-k^*)(k-k^*-1)m(m-1)}{n^3}\frac{||\Delta||_2^2}{||\Sigma||_F} + O_p(\frac{n^{1/2}(\Delta^T \Sigma \Delta)^{1/2}}{||\Sigma||_F}).
\end{align*}
Similarly, we can show when $k^* >  m$
\[\frac{1}{n^3\|\Sigma\|_F}G_k(m) = \frac{1}{n\|\Sigma\|_F}G_k^Y(m) +  \frac{k^*(k^*-1)(k-m)(k-m-1)}{n^3}\frac{||\Delta||_2^2}{||\Sigma||_F} + O_p(\frac{n^{1/2}(\Delta^T \Sigma \Delta)^{1/2}}{||\Sigma||_F}).\]
The last part is converging to 0 in probability. Therefore, the test statistic $T_n$ can be viewed as an extension to the original process. The second terms are also a process depend on $m$ and $k^*$. Under the fixed alternative, the $G_k(m)$ converge to the process
\[\frac{1}{n^3\|\Sigma\|_F}\{G_{\floor{nt}}(\floor{ns}) \}_{s\in[0,1]} \to G(s,t) + b\Lambda(s,t),\]
where
\[\Lambda(s,t) = \begin{cases} (t-r)^2s^2 & s \le r\\
r^2 (t-s)^2  & s > r\\
0 & otherwise
\end{cases}.\]
This implies that, when $b = 0$, the process is the same with the null process, and the proposed monitoring scheme will have trivial power. When the $b$ is not zero, since the remainder term is positive, we will have non -trivial power.

When
\[\frac{n\Delta^T\Delta}{||\Sigma||_F} \to \infty.\]
Following above decomposition, we have
\[
\max_k T_n(k) \ge T_{n}(k^*) = \frac{1}{n\|\Sigma\|_F}D_{nT}^Y(k^*) + O(\frac{n||\Delta||_2^2}{||\Sigma||_F}) \to \infty
\]
Since the first term is pivotal and is bounded in probability, the test have power converging to 1.
\end{proof}

\begin{proof}[Proof of Theorem 3]
We can directly apply the results in Theorem 2.1 and 2.2 in Zhang et al.(2020), which stated that for
\[S_{n,q,c}(r;[a,b]) = \sum_{l = 1}^p\sum_{\floor{na} +1\le i_1, \ldots,i_c\le \floor{nr}}^*  \sum_{\floor{nr}+1 \le j_1, \ldots,j_{q-c}\le \floor{nb}}^*  \prod_{t=1}^cX_{i_t, l}\prod_{g = 1}^{q-c}X_{j_g, l},\]
we have 
\[\frac{1}{\sqrt{n^q \|\Sigma\|_q^q}}S_{n,q,c}(r;[a,b]) \weak Q_{q,c}(r;[a,b]), \]
where $Q_{q,c}$ is the Gaussian process stated in Theorem 4. The monitoring statistic is a continuous transformation of process $S_{n,q,c}$'s and the asymptotic result follows. 
\end{proof}
\begin{proof}[Proof of Theorem 4]
We first discuss the case when $\frac{n^{q/2}\|\Delta\|_q^q}{\|\Sigma\|_q^{q/2}} \to \gamma \in [0 , +\infty)$ and the true change point is at location  $k^* = \floor{nr}$. Here we adopt the process convergence results in Theorem 2.3 of Zhang et al.(2020), which stated that for $(k,m) = (\floor{ns}, \floor{nt})$,
\begin{align*}
 \frac{1}{\sqrt{n^{3q}\|\Sigma\|_q^q}}D_{n,q}(s;[0,b]) &= \frac{1}{\sqrt{n^{3q}\| \Sigma\|_q^q}}\sum_{l = 1}^p\sum_{0\le i_1, \ldots,i_q\le k}^*  \sum_{k+1 \le j_1, \ldots,j_q\le m}^*  (X_{i_1, l} - X_{j_1, l}) \cdots (X_{i_q, l} - X_{j_q, l}),\\
 & \weak G_q(s,t) + \gamma J_q(s;[0,t])
\end{align*}
where 

 \[ J_q(s;[0,t]) = \begin{cases} 
r^q(t-s)^q & r\le s < t\\
s^q(t-r)^q & s \le r < t\\
0 & otherwise
\end{cases}
\]
Therefore, by continuous mapping theorem, when $\gamma \in [0, + \infty)$, the results in the theorem hold. 

For the case $\frac{n^{q/2}\|\Delta\|_q^q}{\|\Sigma\|_q^{q/2}}  \to +\infty$
\[
\max_k T_{n,q}(k) \ge T_{n,q}(k^*) = \frac{1}{n\|\Sigma\|_F}D_{nT}^Y(k^*) + C \frac{n^{q/2}\|\Delta\|_q^q}{\|\Sigma\|_q^{q/2}} \to \infty
\]
\end{proof}
	\begin{proof}[Proof of Theorem 5]
	By straightforward calculation, we have
	\begin{align*}
	    \widehat{\|\Sigma\|_F^2} &= \frac{1}{4{n \choose 4}}\sum_{1 \leq j_1 < j_2 < j_3 < j_4 \leq n}tr\left((X_{j_1} - X_{j_2})(X_{j_1} - X_{j_2})^T(X_{j_3} - X_{j_4})(X_{j_3} - X_{j_4})^T\right)\\
	    &= \frac{1}{4{n \choose 4}}\sum_{1 \leq j_1 < j_2 < j_3 < j_4 \leq n}[(X_{j_1} - X_{j_2})^T(X_{j_3} - X_{j_4})]^2\\
	    &= \frac{1}{4{n \choose 4}}\sum_{1 \leq j_1 < j_2 < j_3 < j_4 \leq n}[(X_{j_1}^TX_{j_3})^2 + (X_{j_2}^TX_{j_3})^2 + (X_{j_2}^TX_{j_4})^2 + (X_{j_1}^TX_{j_4})^2]\\
	    &- \frac{2}{4{n \choose 4}}\sum_{1 \leq j_1 < j_2 < j_3 < j_4 \leq n}[X_{j_1}^TX_{j_3}X_{j_1}^TX_{j_4} + X_{j_2}^TX_{j_3}X_{j_2}^TX_{j_4} +  X_{j_1}^TX_{j_3}X_{j_2}^TX_{j_3} + X_{j_1}^TX_{j_4}X_{j_2}^TX_{j_4}]\\
	    &+ \frac{2}{4{n \choose 4}}\sum_{1 \leq j_1 < j_2 < j_3 < j_4 \leq n}[X_{j_1}^TX_{j_3}X_{j_2}^TX_{j_4} + X_{j_2}^TX_{j_3}X_{j_1}^TX_{j_4}]\\
	    & = I_{n,1}+I_{n,2}+I_{n,3}+I_{n,4} - (I_{n,5}+I_{n,6}+I_{n,7}+I_{n,8}) + (I_{n,9}+I_{n,10}).
	\end{align*}
	
	For $I_{n,1}$,
	$$\E[I_{n,1}] = \frac{1}{4{n \choose 4}}\sum_{1 \leq j_1 < j_2 < j_3 < j_4 \leq n}\E[(X_{j_1}^TX_{j_3})^2] = \frac{1}{4}tr(\E[X_{j_3}X_{j_3}^TX_{j_1}X_{j_1}^T]) = \|\Sigma\|_F^2/4.$$
	Thus $\E[I_{n,1}/\|\Sigma\|_F^2] = 1/4$. By similar arguments, it is obvious to see that $\E[I_{n,i}/\|\Sigma\|_F^2] = 1/4$ for $i = 1,2,3,4$, and $\E[I_{n,i}/\|\Sigma\|_F^2] = 0$ for $i = 5,...,10$.
	
	The outline of the proof is as following. We will show that $4I_{n,i}/\|\Sigma\|_F^2 \rightarrow_p 1$ for $i = 1,2,3,4$, and $I_{n,i}/\|\Sigma\|_F^2 \rightarrow_p 0$, for $i = 5,...,10$. Since some of the $I_{n,i}$ share very similar structures, we will only present the proof for (1) $4I_{n,1}/\|\Sigma\|_F^2 \rightarrow_p 1$ and (2) $I_{n,5}/\|\Sigma\|_F^2 \rightarrow_p 0$. Other terms can be proved by similar arguments. 
	
	To show (1), it suffices to show that $\E[16I_{n,1}^2/\|\Sigma\|_F^4] \rightarrow 1$. To see this,
	\begin{align*}
	    &\E[16I_{n,1}^2/\|\Sigma\|_F^4] = \frac{1}{{n \choose 4}^2\|\Sigma\|_F^4}\sum_{1 \leq j_1 < j_2 < j_3 < j_4 \leq n}\sum_{1 \leq j_5 < j_6 < j_7 < j_8 \leq n}\E[(X_{j_1}^TX_{j_3})^2(X_{j_5}^TX_{j_7})^2]\\
	    &=\frac{1}{{n \choose 4}^2\|\Sigma\|_F^4}\sum_{1 \leq j_1 < j_2 < j_3 < j_4 \leq n}\sum_{1 \leq j_5 < j_6 < j_7 < j_8 \leq n}\sum_{l_1,l_2,l_3,l_4 = 1}^p\E[X_{j_1,l_1}X_{j_3,l_1}X_{j_1,l_2}X_{j_3,l_2}X_{j_5,l_3}X_{j_7,l_3}X_{j_5,l_4}X_{j_7,l_4}].
	\end{align*}
	As we know that the expectation of a product of random variables can be expressed in terms of joint cumulants, we have
	$$\E[X_{j_1,l_1}X_{j_3,l_1}X_{j_1,l_2}X_{j_3,l_2}X_{j_5,l_3}X_{j_7,l_3}X_{j_5,l_4}X_{j_7,l_4}] = \sum_{\pi}\prod_{B \in \pi}cum(X_{j,l}: (j,l) \in B),$$
	where $\pi$ runs through the list of all partitions of $\{(j_1,l_1), (j_1,l_2),...,(j_7,l_3), (j_7,l_4)\}$ and $B$ runs through the list of all blocks of the partition $\pi$. Since $j_1 < j_3$ and $j_5 < j_7$, it is impossible to have three or more indices in $\{j_1,j_3,j_5,j_7\}$ such that they are identical. Thus for the right hand side of the above expression, we only need to take the sum over all partitions with all block sizes smaller than $5$, because for joint cumulants with order greater than $5$, it must contain at least 3 indices from $j_1,j_3,j_5,j_7$ and at least one is not identical to the other two. And the joint cumulants will equal to zero since it involves two or more independent random variables. 
	
	Also since the mean of all random variables included in the left hand side of the above expression are all zero, we do not need to consider the partition with block size 1. Thus the expression can be simplified as 
	\begin{align*}
	    &\E[X_{j_1,l_1}X_{j_3,l_1}X_{j_1,l_2}X_{j_3,l_2}X_{j_5,l_3}X_{j_7,l_3}X_{j_5,l_4}X_{j_7,l_4}] \\
	   = & C_1^{(j_1,j_3,j_5,j_7)}\E[X_{0,l_1}X_{0,l_2}X_{0,l_3}X_{0,l_4}]^2 +  C_2^{(j_1,j_3,j_5,j_7)}\E[X_{0,l_1}X_{0,l_2}X_{0,l_3}X_{0,l_4}]\Sigma_{l_1,l_2}\Sigma_{l_3,l_4}\\
	   & + \Sigma_{l_1,l_2}^2\Sigma_{l_3,l_4}^2,
	\end{align*}
	where $C_1^{(j_1,j_3,j_5,j_7)}$, $C_2^{(j_1,j_3,j_5,j_7)}$ are finite positive constants purely based on the value of $j_1,j_3,j_5,j_7$. $C_1^{(j_1,j_3,j_5,j_7)}$ can only be nonzero if $j_1 = j_5$ and $j_3 = j_7$, and $C_2^{(j_1,j_3,j_5,j_7)}$ is nonzero if at least two of $(j_1, j_3, j_5, j_7)$ are equal. This implies that $$\sum_{1 \leq j_1 < j_2 < j_3 < j_4 \leq n}\sum_{1 \leq j_5 < j_6 < j_7 < j_8 \leq n}C_1^{(j_1,j_3,j_5,j_7)} = o(n^8),$$
	and
	$$\sum_{1 \leq j_1 < j_2 < j_3 < j_4 \leq n}\sum_{1 \leq j_5 < j_6 < j_7 < j_8 \leq n}C_2^{(j_1,j_3,j_5,j_7)} = o(n^8).$$
	
	Furthermore, according to Assumption 2, $\sum_{l_1,l_2,l_3,l_4 = 1}^pcum(X_{0,l_1},X_{0,l_2}, X_{0,l_3}, X_{0,l_4})^2 \leq C\|\Sigma\|_F^4$. It can be verified that  
	\begin{align*}
	    \sum_{l_1,l_2,l_3,l_4 = 1}^p\E[X_{0,l_1}X_{0,l_2}X_{0,l_3}X_{0,l_4}]^2 &\lesssim \sum_{l_1,l_2,l_3,l_4 = 1}^pcum(X_{0,l_1},X_{0,l_2}, X_{0,l_3}, X_{0,l_4})^2 \\
	    & + \sum_{l_1,l_2,l_3,l_4 = 1}^p\Sigma_{l_1,l_2}^2\Sigma_{l_3,l_4}^2\\
	    &\lesssim \|\Sigma\|_F^4,
	\end{align*}
	and by using the Cauchy-Schwartz inequaility,
	\begin{align}
	    &\sum_{l_1,l_2,l_3,l_4 = 1}^p\E[X_{0,l_1}X_{0,l_2}X_{0,l_3}X_{0,l_4}]\Sigma_{l_1,l_2}\Sigma_{l_3,l_4}\nonumber\\
	    \leq &\sqrt{\sum_{l_1,l_2,l_3,l_4 = 1}^p\E[X_{0,l_1}X_{0,l_2}X_{0,l_3}X_{0,l_4}]^2}\sqrt{\sum_{l_1,l_2,l_3,l_4 = 1}^p\Sigma_{l_1,l_2}^2\Sigma_{l_3,l_4}^2} \leq  \sqrt{C}\|\Sigma\|_F^4.\label{eq:bound}
	\end{align}
	This indicates that 
	\begin{align*}
	    &\E[16I_{n,1}^2/\|\Sigma\|_F^4] \\
	    =& \frac{1}{{n \choose 4}^2\|\Sigma\|_F^4}\sum_{1 \leq j_1 < j_2 < j_3 < j_4 \leq n}\sum_{1 \leq j_5 < j_6 < j_7 < j_8 \leq n}C_1^{(j_1,j_3,j_5,j_7)}\sum_{l_1,l_2,l_3,l_4 = 1}^p\E[X_{0,l_1}X_{0,l_2}X_{0,l_3}X_{0,l_4}]^2\\
	    +& \frac{1}{{n \choose 4}^2\|\Sigma\|_F^4}\sum_{1 \leq j_1 < j_2 < j_3 < j_4 \leq n}\sum_{1 \leq j_5 < j_6 < j_7 < j_8 \leq n}C_2^{(j_1,j_3,j_5,j_7)}\sum_{l_1,l_2,l_3,l_4 = 1}^p\E[X_{0,l_1}X_{0,l_2}X_{0,l_3}X_{0,l_4}]\Sigma_{l_1,l_2}\Sigma_{l_3,l_4}\\
	    +& \frac{1}{{n \choose 4}^2\|\Sigma\|_F^4}\sum_{1 \leq j_1 < j_2 < j_3 < j_4 \leq n}\sum_{1 \leq j_5 < j_6 < j_7 < j_8 \leq n}\sum_{l_1,l_2,l_3,l_4 = 1}^p\Sigma_{l_1,l_2}^2\Sigma_{l_3,l_4}^2  = o(1) + o(1) + 1 \rightarrow 1.
	\end{align*}
	Thus, $4I_{n,1}/\|\Sigma\|_F^2 \rightarrow_p 1$, and (1) is proved. By similar arguments, $4I_{n,i}/\|\Sigma\|_F^2 \rightarrow_p 1$ holds for $i = 2,3,4$.
	
	To show (2), we need to prove $\E[I_{n,5}^2/\|\Sigma\|_F^4] \rightarrow 0$. To see this, 
	\begin{align*}
	    &\E[I_{n,5}^2/\|\Sigma\|_F^4] = \frac{1}{4{n \choose 4}^2\|\Sigma\|_F^4}\sum_{1 \leq j_1 < j_2 < j_3 < j_4 \leq n}\sum_{1 \leq j_5 < j_6 < j_7 < j_8 \leq n}\E[(X_{j_1}^TX_{j_3}X_{j_1}^TX_{j_4})(X_{j_5}^TX_{j_7}X_{j_5}^TX_{j_8})]\\
	    =&\frac{1}{{n \choose 4}^2\|\Sigma\|_F^4}\sum_{1 \leq j_1 < j_2 < j_3 < j_4 \leq n}\sum_{1 \leq j_5 < j_6 < j_7 < j_8 \leq n}\sum_{l_1,l_2,l_3,l_4 = 1}^p\E[X_{j_1,l_1}X_{j_3,l_1}X_{j_1,l_2}X_{j_4,l_2}X_{j_5,l_3}X_{j_7,l_3}X_{j_5,l_4}X_{j_8,l_4}].
	\end{align*}
	
	By similar arguments for the joint cumulants we provided in the the proof for (1), it can be proved that
    \begin{align*}
        &\E[X_{j_1,l_1}X_{j_3,l_1}X_{j_1,l_2}X_{j_4,l_2}X_{j_5,l_3}X_{j_7,l_3}X_{j_5,l_4}X_{j_8,l_4}]\\
        =&C_1^{(j_1,j_3,j_4,j_5,j_7,j_8)}\E[X_{0,l_1}X_{0,l_2}X_{0,l_3}X_{0,l_4}]\Sigma_{l_1,l_3}\Sigma_{l_2,l_4} + C_2^{(j_1,j_3,j_4,j_5,j_7,j_8)}\Sigma_{l_1,l_2}\Sigma_{l_3,l_4}\Sigma_{l_1,l_3}\Sigma_{l_2,l_4}.
    \end{align*}
    
    If $C_1^{(j_1,j_3,j_4,j_5,j_7,j_8)} \neq 0$, then $j_1 = j_5$. And if $C_2^{(j_1,j_3,j_4,j_5,j_7,j_8)} \neq 0$, $j_3 = j_5$ and $j_4 = j_8$. These two properties guarantee that 
    $$\sum_{1 \leq j_1 < j_2 < j_3 < j_4 \leq n}\sum_{1 \leq j_5 < j_6 < j_7 < j_8 \leq n}C_1^{(j_1,j_3,j_4,j_5,j_7,j_8)} = o(n^8),$$
    and
    $$\sum_{1 \leq j_1 < j_2 < j_3 < j_4 \leq n}\sum_{1 \leq j_5 < j_6 < j_7 < j_8 \leq n}C_2^{(j_1,j_3,j_4,j_5,j_7,j_8)} = o(n^8).$$
    Furthermore we have shown the bound for  $\sum_{l_1,l_2,l_3,l_4 = 1}^p\E[X_{0,l_1}X_{0,l_2}X_{0,l_3}X_{0,l_4}]\Sigma_{l_1,l_2}\Sigma_{l_3,l_4}$ in (\ref{eq:bound}). And
    \begin{align*}
        &\sum_{l_1,l_2,l_3,l_4 = 1}^p\Sigma_{l_1,l_2}\Sigma_{l_3,l_4}\Sigma_{l_1,l_3}\Sigma_{l_2,l_4} = \sum_{l_1,l_4 = 1}^p\left(\sum_{l_2 = 1}^p\Sigma_{l_1,l_2}\Sigma_{l_2,l_4}\right)\left(\sum_{l_3=1}^p\Sigma_{l_1,l_3}\Sigma_{l_3,l_4}\right)\\
        =&\sum_{l_1,l_4 = 1}^p[(\Sigma^2)_{l_1,l_4}]^2 = tr(\Sigma^4) = o(\|\Sigma\|_F^4),
    \end{align*}
    by Assumption 1. Thus,
    \begin{align*}
	    &\E[I_{n,5}^2/\|\Sigma\|_F^4] \\
	    =& \frac{1}{4{n \choose 4}^2\|\Sigma\|_F^4}\sum_{1 \leq j_1 < j_2 < j_3 < j_4 \leq n}\sum_{1 \leq j_5 < j_6 < j_7 < j_8 \leq n}C_1^{(j_1,j_3,j_4,j_5,j_7,j_8)}\sum_{l_1,l_2,l_3,l_4 = 1}^p\E[X_{0,l_1}X_{0,l_2}X_{0,l_3}X_{0,l_4}]\Sigma_{l_1,l_3}\Sigma_{l_2,l_4}\\
	    +&\frac{1}{4{n \choose 4}^2\|\Sigma\|_F^4}\sum_{1 \leq j_1 < j_2 < j_3 < j_4 \leq n}\sum_{1 \leq j_5 < j_6 < j_7 < j_8 \leq n} C_2^{(j_1,j_3,j_4,j_5,j_7,j_8)}\sum_{l_1,l_2,l_3,l_4 = 1}^p\Sigma_{l_1,l_2}\Sigma_{l_3,l_4}\Sigma_{l_1,l_3}\Sigma_{l_2,l_4}\\
	    =&o(1) + o(1) \rightarrow 1.
	\end{align*}
	This indicates $I_{n,5}/\|\Sigma\|_F^2 \rightarrow_p 0$. And by similar arguments we can prove that $I_{n,i}/\|\Sigma\|_F^2 \rightarrow_p 0$, for all $i = 6,...,10$. Combine the above results, we have $\widehat{\|\Sigma\|_F^2}/\|\Sigma\|_F^2 \rightarrow_p 1$. This completes the proof.

	\end{proof}

	\begin{proof}[Proof of Theorem 6]
	We can rewrite $\widehat{\|\Sigma\|_q^q}$ as
	\begin{align*}
	    &\widehat{\|\Sigma\|_q^q} = \frac{1}{2^q{n \choose 2q}}\sum_{l_1,l_2 = 1}^p\sum_{1 \leq i_1 < \cdots < i_q < j_1 < \cdots < j_q \leq n} \prod_{k = 1}^q (X_{i_k,l_1}X_{i_k,l_2} + X_{j_k,l_1}X_{j_k,l_2} - X_{i_k,l_1}X_{j_k,l_2} - X_{j_k,l_1}X_{i_k,l_2})\\
	    =&\frac{1}{2^q{n \choose 2q}}\sum_{1 \leq i_1 < \cdots < i_q < j_1 < \cdots < j_q \leq n} \sum_{t_1,s_1 \in \{i_1,j_1\}}\cdots\sum_{t_q,s_q \in \{i_q,j_q\}}\sum_{l_1,l_2 = 1}^p\prod_{k = 1}^q (-1)^{\1{t_k \neq s_k}}X_{t_k,l_1}X_{s_k,l_2}\\
	    =& \frac{1}{2^q{n \choose 2q}}\sum_{1 \leq i_1 < \cdots < i_q < j_1 < \cdots < j_q \leq n}\sum_{t_1 \in \{i_1,j_1\}}\cdots\sum_{t_q \in \{i_q,j_q\}}\sum_{l_1,l_2 = 1}^p \prod_{k = 1}^q X_{t_k,l_1}X_{t_k,l_2} \\
	    &+ \frac{1}{2^q{n \choose 2q}}\sum_{1 \leq i_1 < \cdots < i_q < j_1 < \cdots < j_q \leq n} \sum_{t_1,s_1 \in \{i_1,j_1\}}\cdots\sum_{t_q,s_q \in \{i_q,j_q\}}\\
	    & \qquad 
	    \qquad \sum_{l_1,l_2 = 1}^p\1{\cup_{k = 1}^q \{t_k \neq s_k\}}\prod_{k = 1}^q (-1)^{\1{t_k \neq s_k}}X_{t_k,l_1}X_{s_k,l_2}.
	\end{align*}
	The second equality in the above expression is by calculating the cross products among $q$ brackets, and the third equality is splitting the terms based on different values of $t_k,s_k$ for $k = 1,...,q$. The first term in the third equality contains all products with $t_k = s_k$ for all $k = 1,...,q$, and the second term contains products with at least one $k = 1,...,q$ such that $t_k \neq s_k$.
	
	The outline of the proof is as follows. We want to show:
	\begin{enumerate}
	    \item for every $t_1 \in \{i_1,j_1\},...,t_q \in \{i_q,j_q\}$,
	    $$I({t_1,...,t_q}) = \frac{1}{{n \choose 2q}\|\Sigma\|_q^q}\sum_{1 \leq i_1 < \cdots < i_q < j_1 < \cdots < j_q \leq n}\sum_{l_1,l_2 = 1}^p\prod_{k = 1}^qX_{t_k,l_1}X_{t_k,l_2} \rightarrow_p 1;$$
	    \item for every $t_1,s_1 \in \{i_1,j_1\},...,t_q,s_q \in \{i_q,j_q\}$ and there exists at least one $k = 1,...,q$ such that $t_k \neq s_k$,
	    $$J({t_1,s_1,...,t_q,s_q}) = \frac{1}{{n \choose 2q}\|\Sigma\|_q^q}\sum_{1 \leq i_1 < \cdots < i_q < j_1 < \cdots < j_q \leq n}\sum_{l_1,l_2 = 1}^p\prod_{k = 1}^qX_{t_k,l_1}X_{s_k,l_2}\rightarrow_p 0.$$
	\end{enumerate}
	And it is easy to see that if these two results hold, then $\widehat{\|\Sigma\|_q^q}/\|\Sigma\|_q^q \rightarrow_p 1$. As we observe that most of terms are structurally very similar, we shall only present the proof for $I({i_1,...,i_q}) \rightarrow_p 1$ and a general proof for (2). 
	
	It is trivial to see that $\E[\sum_{l_1,l_2 = 1}^p\prod_{k = 1}^qX_{t_k,l_1}X_{t_k,l_2}/\|\Sigma\|_q^q] = 1$. This indicates that to show (1), it suffices to show that $\E[I(i_1,...,i_q)^2] \rightarrow 1$. To show this,
	
	\begin{align*}
	    \E[I(i_1,...,i_q)^2] =& \frac{1}{{n \choose 2q}^2\|\Sigma\|_q^{2q}}\sum_{1 \leq i_1 < \cdots < i_q < j_1 < \cdots < j_q \leq n}\sum_{1 \leq i_1' < \cdots < i_q' < j_1' < \cdots < j_q' \leq n}\\
	    &\qquad\sum_{l_1,l_2,l_3,l_4 = 1}^p\E\left[\prod_{k = 1}^qX_{i_k,l_1}X_{i_k,l_2}X_{i_k',l_3}X_{i_k',l_4}\right].
	\end{align*}
	
	Due to the special structure of our statistic, 
	$$\E\left[\prod_{k = 1}^qX_{i_k,l_1}X_{i_k,l_2}X_{i_k',l_3}X_{i_k',l_4}\right] = \sum_{m = 0}^{q}C_m E(X_{0,l_1}X_{0,l_2}X_{0,l_3}X_{0,l_4})^m (\Sigma_{l_1,l_2}\Sigma_{l_3,l_4})^{q-m},$$
	where $C_m = C_m(i_1,...,i_q,i_1',...,i_q') \geq 0$ is a function of all indices for all $m = 1,2,...,q$. $C_m = 1$ if there are exact $m$ indices in $\{i_1,...,i_q\}$ which equal to $m$ indices in $\{i_1',...,i_q'\}$, and $C_m = 0$ otherwise. These events are mutually exclusive which indicates that $\sum_{m = 0}^qC_m = 1$. This indicates that for all $m = 1,...,q$, 
	$$\sum_{1 \leq i_1 < \cdots < i_q < j_1 < \cdots < j_q \leq n}\sum_{1 \leq i_1' < \cdots < i_q' < j_1' < \cdots < j_q' \leq n}C_m(i_1,...,i_q,i_1',...,i_q') = o(n^{4q}),$$
	and 
	$$\frac{1}{{n \choose 2q}^2}\sum_{1 \leq i_1 < \cdots < i_q < j_1 < \cdots < j_q \leq n}\sum_{1 \leq i_1' < \cdots < i_q' < j_1' < \cdots < j_q' \leq n}C_0(i_1,...,i_q,i_1',...,i_q') \rightarrow 1.$$
	Furthermore, for any $m = 1,...,q$, by H\"older's inequaility for vector spaces, we have 
	\begin{align*}
	    &\sum_{l_1,l_2,l_3,l_4 = 1}^p|E(X_{0,l_1}X_{0,l_2}X_{0,l_3}X_{0,l_4})|^m |\Sigma_{l_1,l_2}\Sigma_{l_3,l_4}|^{q-m}\\
	    \leq &  \left(\sum_{l_1,l_2,l_3,l_4 = 1}^p|E(X_{0,l_1}X_{0,l_2}X_{0,l_3}X_{0,l_4})^m|^{q/m}\right)^{m/q}\left(\sum_{l_1,l_2,l_3,l_4 = 1}^p(|\Sigma_{l_1,l_2}\Sigma_{l_3,l_4}|^{q-m})^{q/(q-m)}\right)^{(q-m)/q}\\
	    =& \left(\sum_{l_1,l_2,l_3,l_4 = 1}^p|E(X_{0,l_1}X_{0,l_2}X_{0,l_3}X_{0,l_4})|^q\right)^{m/q}\left(\sum_{l_1,l_2,l_3,l_4 = 1}^p|\Sigma_{l_1,l_2}|^q|\Sigma_{l_3,l_4}|^q\right)^{(q-m)/q} \\
	    \leq&  C\|\Sigma\|_q^{2m} \|\Sigma\|_q^{2(q-m)} = C\|\Sigma\|_q^{2q},
	\end{align*}
	where the last inequality is due to Assumption 5, and to see this, 
	\begin{align*}
	    &\sum_{l_1,l_2,l_3,l_4 = 1}^p|E(X_{0,l_1}X_{0,l_2}X_{0,l_3}X_{0,l_4})|^q \\
	    \leq&  C\sum_{l_1,l_2,l_3,l_4 = 1}^p|cum(X_{0,l_1},X_{0,l_2},X_{0,l_3},X_{0,l_4})|^q + |\Sigma_{l_1,l_2}\Sigma_{l_3,l_4}|^q + |\Sigma_{l_1,l_3}\Sigma_{l_2,l_4}|^q +|\Sigma_{l_1,l_4}\Sigma_{l_2,l_3}|^q\\
	    \leq& C\sum_{1 \leq l_1 \leq l_2 \leq l_3 \leq l_4 \leq p}(1 \vee (l_4 - l_1))^{-2rq} + 3C\|\Sigma\|_q^{2q} \leq Cp^2 + 3C\|\Sigma\|_q^{2q} \leq C\|\Sigma\|_q^{2q},
	\end{align*}
	for some generic positive constant $C$, since $\|\Sigma\|_q^{2q} = (\sum_{i,j = 1}^p\Sigma_{i,j}^q)^2 \geq Cp^2$ under Assumption (5.1).
	Therefore, 
	\begin{align*}
	     &\E[I(i_1,...,i_q)^2] \\
	     =& \frac{1}{{n \choose 2q}^2\|\Sigma\|_q^{2q}}\sum_{1 \leq i_1 < \cdots < i_q < j_1 < \cdots < j_q \leq n}\sum_{1 \leq i_1' < \cdots < i_q' < j_1' < \cdots < j_q' \leq n} \sum_{l_1,l_2,l_3,l_4 = 1}^q\E\left[\prod_{k = 1}^pX_{i_k,l_1}X_{i_k,l_2}X_{i_k',l_3}X_{i_k',l_4}\right]\\
	    =& \frac{1}{{n \choose 2q}^2\|\Sigma\|_q^{2q}}\sum_{1 \leq i_1 < \cdots < i_q < j_1 < \cdots < j_q \leq n}\sum_{1 \leq i_1' < \cdots < i_q' < j_1' < \cdots < j_q' \leq n}\\
	    &\qquad\sum_{l_1,l_2,l_3,l_4 = 1}^p\sum_{m = 0}^{q}C_m \E(X_{0,l_1}X_{0,l_2}X_{0,l_3}X_{0,l_4})^m (\Sigma_{l_1,l_2}\Sigma_{l_3,l_4})^{q-m}\\
	    =&\frac{1}{{n \choose 2q}^2\|\Sigma\|_q^{2q}}\sum_{1 \leq i_1 < \cdots < i_q < j_1 < \cdots < j_q \leq n}\sum_{1 \leq i_1' < \cdots < i_q' < j_1' < \cdots < j_q' \leq n}C_0(i_1,...,i_q,i_1',...,i_q')\sum_{l_1,l_2,l_3,l_4 = 1}^p\Sigma_{l_1,l_2}^q\Sigma_{l_3,l_4}^q\\
	    +&\frac{1}{{n \choose 2q}^2\|\Sigma\|_q^{2q}}o(n^{4q})\sum_{l_1,l_2,l_3,l_4 = 1}^p\sum_{m = 1}^{q} \E(X_{0,l_1}X_{0,l_2}X_{0,l_3}X_{0,l_4})^m (\Sigma_{l_1,l_2}\Sigma_{l_3,l_4})^{q-m}\\
	    =&1 + o(1) \rightarrow 1.
	\end{align*}
	This completes the proof for (1). To show (2), it suffices to show that $\E[J(t_1,s_1,...,t_q,s_q)^2] \rightarrow 0$. Specifically,   
	\begin{align*}
	    \E[J(t_1,s_1,...,t_q,s_q)^2]=& \frac{1}{{n \choose 2q}^2\|\Sigma\|_q^{2q}}\sum_{1 \leq i_1 < \cdots < i_q < j_1 < \cdots < j_q \leq n}\sum_{1 \leq i_1' < \cdots < i_q' < j_1' < \cdots < j_q' \leq n}\\
	    &\qquad\sum_{l_1,l_2,l_3,l_4 = 1}^p\E\left[\prod_{k=1}^qX_{t_k,l_1}X_{s_k,l_2}X_{t_k',l_3}X_{s_k',l_4}\right],
	\end{align*}
	for $t_1,s_1 \in \{i_1,j_1\},...,t_q,s_q \in \{i_q,j_q\},t_1',s_1' \in \{i_1',j_1'\},...,t_q',s_q' \in \{i_q',j_q'\}$, and there exists at least one $k = 1,...,q$ such that $t_k \neq s_k$ ($t_k' \neq s_k')$. Since the expectation of a product of random variables can be expressed in terms of joint cumulants, we have
	$$\E\left[\prod_{k = 1}^qX_{t_k,l_1}X_{s_k,l_2}X_{t_k',l_3}X_{s_k',l_4}\right] = \sum_{\pi}\prod_{B \in \pi}cum(X_{i,l}: (i,l) \in B),$$
	where $\pi$ runs through the list of all partitions of $\{(t_1,l_1),(s_1,l_2),...,(t_q',l_3),(s_q', l_4)\}$ and $B$ runs thorough the list of all blocks of the partition $\pi$. Due to the special structure of our statistic, there is a set of partitions $\mathcal{S}$ such that for every $\pi \in \mathcal{S}$, the product of joint cumulants over all $B \in \pi$ is zero. And for each $\pi \in \mathcal{S}^c$ there are nice properties related to the blocks $B \in \pi$. Here we shall illustrate these properties as follows. To be clear, since we are dealing with a double indexed array $X_{i,l}$, we call $``i"$ as the temporal index and $``l"$ as the spatial index. For $\forall \pi \in \mathcal{S}^c$,
	\begin{enumerate}
	    \item {\bf The size of every block $B \in \pi$ cannot exceed 4.} Since $i_1,...,i_q,j_1,...,j_q$ are all distinct, and $i_1',...,i_q',j_1',...,j_q'$ are all distinct, it is impossible to have any three indices in $\{i_1,...,i_q,j_1,...,j_q,i_1',...,i_q',j_1',...,j_q'\}$ that are equal. And any joint cumulants of order greater than or equal to 5 will include at least three indices and they cannot be all equal.
	    
	    \item {\bf There are no blocks with size 1.} This is because the cumulant of a single random variable with mean zero is also zero. 
	    
	    \item {\bf Every $B \in \pi$ must contain only one distinct temporal index.} Otherwise $\prod_{B \in \pi}cum(X_{i,l}: (i,l) \in B) = 0$.
	\end{enumerate}
	The above properties imply that for  $\forall \pi \in \mathcal{S}^c$ and $\forall B \in \pi$, ${cum(X_{i,l}: (i,l) \in B)}$ has to be one of the followings: $cum(X_{0,l_1},X_{0,l_2},X_{0,l_3},X_{0,l_4})$, $cum(X_{0,l_1},X_{0,l_2},X_{0,l_3})$, $cum(X_{0,l_1},X_{0,l_2},X_{0,l_4})$, $cum(X_{0,l_1},X_{0,l_3},X_{0,l_4})$, $cum(X_{0,l_2},X_{0,l_3},X_{0,l_4})$, $\Sigma_{l_1,l_2}$, $\Sigma_{l_1,l_3}$, $\Sigma_{l_1,l_4}$, $\Sigma_{l_2,l_3}$, $\Sigma_{l_2,l_4}$, $\Sigma_{l_3,l_4}$.
	
	If we assume $l_1 \leq l_2 \leq l_3 \leq l_4$, it can be shown that
	\begin{equation}\label{eq:bound}
    \prod_{B \in \pi}cum(X_{i,l}: (i,l) \in B) \leq C(1 \vee (l_2 - l_1))^{-r}(1 \vee (l_4 - l_3))^{-r},
	\end{equation}
	for some generic positive constant $C$ and any partition $\pi$. To see this, we notice that at least one $k = 1,...,q$, say $k_0$, such that $t_{k_0} \neq s_{k_0}$ and $t_{k_0}' \neq s_{k_0}'$. For every $\pi \in \mathcal{S}^c$ there exists $B_1,B_2 \in \pi$ such that $(t_{k_0},l_1) \in B_1$ and $(s_{k_0}',l_4) \in B_2$. Based on the third property above, all other elements in $B_1$ must have the same temporal index as $t_{k_0}$. And because of the first property above, all $i_k$, $j_k$ for $k \neq k_0$ and $s_{k_0}$ are different from $t_{k_0}$. This implies that the spatial indices for all other elements in $B_1$ have to be either $l_3$ or $l_4$, not $l_1$ and $l_2$. For the same reason, the spatial indices for all other elements in $B_2$ can only be either $l_1$ or $l_2$. Therefore,
	$$cum(X_{i,l}: (i,l) \in B_1) \in \{cum(X_{0,l_1},X_{0,l_3},X_{0,l_4}), \Sigma_{l_1,l_3},\Sigma_{l_1,l_4}\},$$ and $$cum(X_{i,l}: (i,l) \in B_2) \in \{cum(X_{0,l_1},X_{0,l_2},X_{0,l_4}), \Sigma_{l_1,l_4},\Sigma_{l_2,l_4}\}.$$ Under Assumption (5.2), $cum(X_{i,l}: (i,l) \in B_1) \leq C(1\vee (l_2-l_1))^{-r}$ and  $cum(X_{i,l}: (i,l) \in B_2) \leq C(1\vee (l_4-l_3))^{-r}$. And the joint cumulants are uniformly bounded above for those $B \in \pi \setminus \{B_1,B_2\}$. Thus Equation \ref{eq:bound} is proved.
	
	Furthermore, define $Ind(t_1,s_1,...,t_q,s_q,t_1',s_1',...,t_q',s_q')$ as the indicator function corresponding to the event that for every $k = 1,2,..,q$ that $t_k \neq s_k$, there exists $k' = 1,...,q$ such that $t_k = t_{k'}'$ or $t_k = s_{k'}'$, then $\E\left[\prod_{k = 1}^pX_{t_k,l_1}X_{s_k,l_2}X_{t_k',l_3}X_{s_k',l_4}\right] \neq 0$ only if 
	$$Ind(t_1,s_1,...,t_q,s_q,t_1',s_1',...,t_q',s_q') = 1.$$ It is also easy to see that $$\sum_{1\leq i_1 < \cdots < i_q < j_1 < \cdots < j_q \leq n}\sum_{1\leq i_1' < \cdots < i_q' < j_1' < \cdots < j_q' \leq n}Ind(t_1,s_1,...,t_q,s_q,t_1',s_1',...,t_q',s_q') = o(n^{4q}).$$
	
	Combining all the results above, we have
	\begin{align*}
	    &\E[J(t_1,s_1,...,t_q,s_q)^2]\\
	    =& \frac{1}{{n \choose 2q}^2\|\Sigma\|_q^{2q}}\sum_{1 \leq i_1 < \cdots < i_q < j_1 < \cdots < j_q \leq n}\sum_{1 \leq i_1' < \cdots < i_q' < j_1' < \cdots < j_q' \leq n}\sum_{l_1,l_2,l_3,l_4 = 1}^p\E\left[\prod_{k=1}^qX_{t_k,l_1}X_{s_k,l_2}X_{t_k',l_3}X_{s_k',l_4}\right],\\
	    \leq& \frac{C}{{n \choose 2q}^2\|\Sigma\|_q^{2q}}\sum_{1 \leq i_1 < \cdots < i_q < j_1 < \cdots < j_q \leq n}\sum_{1 \leq i_1' < \cdots < i_q' < j_1' < \cdots < j_q' \leq n}\\
	    &\qquad\sum_{1 \leq l_1 \leq l_2 \leq l_3 \leq l_4 \leq p}   Ind(t_1,s_1,...,t_q,s_q,t_1',s_1',...,t_q',s_q') (1 \vee (l_2 - l_1))^{-r}(1 \vee (l_4 - l_3))^{-r}\\
	    \leq&  \frac{o(n^{4q})}{{n \choose 2q}^2\|\Sigma\|_q^{2q}}\left(\sum_{1 \leq l_1 \leq l_2 \leq p}(1 \vee (l_2 - l_1))^{-r})\right)^2 \lesssim \frac{p^2}{\|\Sigma\|_q^{2q}}o(1) = o(1) \rightarrow 0,
	\end{align*}
	where the last equality is because $\|\Sigma\|_q^{2q} = (\sum_{i,j = 1}^p\Sigma_{i,j}^q)^2 \gtrsim p^2$.
	
	This completes the proof of (2), as well as the whole proof.

	\end{proof}
Table \ref{tab:size_200} shows additional simulation results for the size of the proposed monitoring statistics for $n=200$. The size distortion problem has improved for almost all settings.
\begin{table}[!h]
\caption{Size of different monitoring procedures}
\label{tab:size_200}
\begin{tabular}{lrrrrrrrrr}
\hline
$(n,p) = (200,200)$            & \multicolumn{3}{c}{T1}                                                     & \multicolumn{3}{c}{T2}                                                     & \multicolumn{3}{c}{T3}                                                     \\ \cline{2-10} 
size $\alpha = 0.1$ & \multicolumn{1}{c}{L2} & \multicolumn{1}{c}{L6} & \multicolumn{1}{c}{Comb} & \multicolumn{1}{c}{L2} & \multicolumn{1}{c}{L6} & \multicolumn{1}{c}{Comb} & \multicolumn{1}{c}{L2} & \multicolumn{1}{c}{L6} & \multicolumn{1}{c}{Comb} \\ \hline
$\rho = 0.2$         & 0.104                  & 0.072                  & 0.074                    & 0.097                  & 0.072                  & 0.073                    & 0.102                  & 0.071                  & 0.073                    \\
$\rho = 0.5$         & 0.105                  & 0.064                  & 0.091                    & 0.107                  & 0.064                  & 0.085                    & 0.104                  & 0.065                  & 0.087                    \\
$\rho = 0.8$         & 0.127                  & 0.037                  & 0.089                    & 0.133                  & 0.038                  & 0.099                    & 0.131                  & 0.039                  & 0.099                    \\ \hline
\end{tabular}
\end{table}
\par
\section*{Acknowledgements}
Xiaofeng Shao acknowleges partial support from  NSF grants DMS-1807032 and DMS-2014018. Hao Yan acknowleges partial support from NSF grants DMS-1830363 and	CMMI-1922739.
\par


\bibhang=1.7pc
\bibsep=2pt
\fontsize{9}{14pt plus.8pt minus .6pt}\selectfont
\renewcommand\bibname{\large \bf References}
\expandafter\ifx\csname
natexlab\endcsname\relax\def\natexlab#1{#1}\fi
\expandafter\ifx\csname url\endcsname\relax
  \def\url#1{\texttt{#1}}\fi
\expandafter\ifx\csname urlprefix\endcsname\relax\def\urlprefix{URL}\fi

  \bibliographystyle{chicago}      
  \bibliography{bibfile}   

\begin{thebibliography}{}

\bibitem[\protect\citeauthoryear{Aminikhanghahi and Cook}{Aminikhanghahi and
  Cook}{2017}]{aminikhanghahi2017survey}
Aminikhanghahi, S. and D.~J. Cook (2017).
\newblock A survey of methods for time series change point detection.
\newblock {\em Knowledge and Information Systems\/}~{\em 51\/}(2), 339--367.

\bibitem[\protect\citeauthoryear{Aue, H{\"o}rmann, Horv{\'a}th,
  Hu{\v{s}}kov{\'a}, and Steinebach}{Aue et~al.}{2012}]{aue2012sequential}
Aue, A., S.~H{\"o}rmann, L.~Horv{\'a}th, M.~Hu{\v{s}}kov{\'a}, and J.~G.
  Steinebach (2012).
\newblock Sequential testing for the stability of high-frequency portfolio
  betas.
\newblock {\em Econometric Theory\/}~{\em 28\/}(4), 804--837.

\bibitem[\protect\citeauthoryear{Aue and Horv{\'a}th}{Aue and
  Horv{\'a}th}{2013}]{aue2013structural}
Aue, A. and L.~Horv{\'a}th (2013).
\newblock Structural breaks in time series.
\newblock {\em Journal of Time Series Analysis\/}~{\em 34\/}(1), 1--16.

\bibitem[\protect\citeauthoryear{Brown, Durbin, and Evans}{Brown
  et~al.}{1975}]{brown1975techniques}
Brown, R.~L., J.~Durbin, and J.~M. Evans (1975).
\newblock Techniques for testing the constancy of regression relationships over
  time.
\newblock {\em Journal of the Royal Statistical Society: Series B (Statistical
  Methodology)\/}~{\em 37\/}(2), 149--163.

\bibitem[\protect\citeauthoryear{B{\"u}cher, Fermanian, and
  Kojadinovic}{B{\"u}cher et~al.}{2019}]{bucher2019combining}
B{\"u}cher, A., J.-D. Fermanian, and I.~Kojadinovic (2019).
\newblock Combining cumulative sum change-point detection tests for assessing
  the stationarity of univariate time series.
\newblock {\em Journal of Time Series Analysis\/}~{\em 40\/}(1), 124--150.

\bibitem[\protect\citeauthoryear{Chen and Qin}{Chen and Qin}{2010}]{chenqin10}
Chen, S. and Y.~Qin (2010).
\newblock A two sample test for high dimensional data with application to
  gene-set testing.
\newblock {\em The Annals of Statistics\/}~{\em 38}, 808--835.

\bibitem[\protect\citeauthoryear{Cho and Fryzlewicz}{Cho and
  Fryzlewicz}{2015}]{cho2015multiple}
Cho, H. and P.~Fryzlewicz (2015).
\newblock Multiple-change-point detection for high dimensional time series via
  sparsified binary segmentation.
\newblock {\em Journal of the Royal Statistical Society: Series B (Statistical
  Methodology)\/}~{\em 77\/}(2), 475--507.

\bibitem[\protect\citeauthoryear{Chu, Stinchcombe, and White}{Chu
  et~al.}{1996}]{chu:96}
Chu, C.-S.~J., M.~Stinchcombe, and H.~White (1996).
\newblock Monitoring structural change.
\newblock {\em Econometrica: Journal of the Econometric Society\/}~{\em 64},
  1045--1065.

\bibitem[\protect\citeauthoryear{Dette and G{\"o}smann}{Dette and
  G{\"o}smann}{2019}]{dette2019likelihood}
Dette, H. and J.~G{\"o}smann (2019).
\newblock A likelihood ratio approach to sequential change point detection for
  a general class of parameters.
\newblock {\em Journal of the American Statistical Association\/}, 1--17.

\bibitem[\protect\citeauthoryear{Enikeeva and Harchaoui}{Enikeeva and
  Harchaoui}{2019}]{enikeeva2019high}
Enikeeva, F. and Z.~Harchaoui (2019).
\newblock High-dimensional change-point detection under sparse alternatives.
\newblock {\em The Annals of Statistics\/}~{\em 47\/}(4), 2051--2079.

\bibitem[\protect\citeauthoryear{Fremdt}{Fremdt}{2015}]{fremdt2015page}
Fremdt, S. (2015).
\newblock Page's sequential procedure for change-point detection in time series
  regression.
\newblock {\em Statistics\/}~{\em 49\/}(1), 128--155.

\bibitem[\protect\citeauthoryear{G\"osmann, Kley, and Dette}{G\"osmann
  et~al.}{2020}]{gkd2020JTSA}
G\"osmann, J., T.~Kley, and H.~Dette (2020).
\newblock A new approach for open-end sequential change point monitoring.
\newblock {\em Journal of Time Series Analysis, to appear\/}.

\bibitem[\protect\citeauthoryear{He, Xu, Wu, and Pan}{He
  et~al.}{2018}]{he2018asymptotically}
He, Y., G.~Xu, C.~Wu, and W.~Pan (2018).
\newblock Asymptotically independent u-statistics in high-dimensional testing.
\newblock {\em arXiv preprint arXiv:1809.00411\/}.

\bibitem[\protect\citeauthoryear{Horv{\'a}th and Hu{\v{s}}kov{\'a}}{Horv{\'a}th
  and Hu{\v{s}}kov{\'a}}{2012}]{horvath2012change}
Horv{\'a}th, L. and M.~Hu{\v{s}}kov{\'a} (2012).
\newblock Change-point detection in panel data.
\newblock {\em Journal of Time Series Analysis\/}~{\em 33\/}(4), 631--648.

\bibitem[\protect\citeauthoryear{Horv{\'a}th, Hu{\v{s}}kov{\'a}, Kokoszka, and
  Steinebach}{Horv{\'a}th et~al.}{2004}]{horvath2004monitoring}
Horv{\'a}th, L., M.~Hu{\v{s}}kov{\'a}, P.~Kokoszka, and J.~Steinebach (2004).
\newblock Monitoring changes in linear models.
\newblock {\em Journal of Statistical Planning and Inference\/}~{\em 126\/}(1),
  225--251.

\bibitem[\protect\citeauthoryear{Jirak}{Jirak}{2015}]{jirak2015uniform}
Jirak, M. (2015).
\newblock Uniform change point tests in high dimension.
\newblock {\em The Annals of Statistics\/}~{\em 43\/}(6), 2451--2483.

\bibitem[\protect\citeauthoryear{Kirch, Muhsal, and Ombao}{Kirch
  et~al.}{2015}]{kirch2015detection}
Kirch, C., B.~Muhsal, and H.~Ombao (2015).
\newblock Detection of changes in multivariate time series with application to
  eeg data.
\newblock {\em Journal of the American Statistical Association\/}~{\em
  110\/}(511), 1197--1216.

\bibitem[\protect\citeauthoryear{Lai}{Lai}{1995}]{lai:95}
Lai, T.~L. (1995).
\newblock Sequential changepoint detection in quality control and dynamical
  systems.
\newblock {\em Journal of Royal Statistical Society, Series B (Statistical
  Methodology)\/}~{\em 57}, 613--658.

\bibitem[\protect\citeauthoryear{Lai}{Lai}{2001}]{lai:01}
Lai, T.~L. (2001).
\newblock Sequential analysis: Some classical problems and new challenges.
\newblock {\em Statistica Sinica\/}~{\em 11}, 303--408.

\bibitem[\protect\citeauthoryear{Lei, Zhang, and Jin}{Lei
  et~al.}{2010}]{lei2010automatic}
Lei, Y., Z.~Zhang, and J.~Jin (2010).
\newblock Automatic tonnage monitoring for missing part detection in
  multi-operation forging processes.
\newblock {\em Journal of Manufacturing Science and Engineering\/}~{\em
  132\/}(5).

\bibitem[\protect\citeauthoryear{L{\'e}vy-Leduc and Roueff}{L{\'e}vy-Leduc and
  Roueff}{2009}]{levy2009detection}
L{\'e}vy-Leduc, C. and F.~Roueff (2009).
\newblock Detection and localization of change-points in high-dimensional
  network traffic data.
\newblock {\em The Annals of Applied Statistics\/}~{\em 3\/}(2), 637--662.

\bibitem[\protect\citeauthoryear{Li}{Li}{2020}]{li2020}
Li, J. (2020).
\newblock Efficient global monitoring statistics for high-dimensional data.
\newblock {\em Quality Reliability Engineering International\/}~{\em 36},
  18--32.

\bibitem[\protect\citeauthoryear{Liu, Zhang, and Mei}{Liu
  et~al.}{2019}]{liu:19}
Liu, K., R.~Zhang, and Y.~Mei (2019).
\newblock Scalable sum-shrinkage schemes for distributed monitoring large-scale
  data streams.
\newblock {\em Statistica Sinica\/}~{\em 29}, 1--22.

\bibitem[\protect\citeauthoryear{Lorden}{Lorden}{1971}]{lorden1971}
Lorden, G. (1971).
\newblock Procedures for reacting to a change in distribution.
\newblock {\em Annals of Mathematical Statistics\/}~{\em 42}, 1897--1908.

\bibitem[\protect\citeauthoryear{MacNeill}{MacNeill}{1974}]{macneill1974tests}
MacNeill, I.~B. (1974).
\newblock Tests for change of parameter at unknown times and distributions of
  some related functionals on brownian motion.
\newblock {\em The Annals of Statistics\/}~{\em 2\/}(5), 950--962.

\bibitem[\protect\citeauthoryear{Matteson and James}{Matteson and
  James}{2014}]{matteson2014nonparametric}
Matteson, D.~S. and N.~A. James (2014).
\newblock A nonparametric approach for multiple change point analysis of
  multivariate data.
\newblock {\em Journal of the American Statistical Association\/}~{\em
  109\/}(505), 334--345.

\bibitem[\protect\citeauthoryear{Mei}{Mei}{2010}]{mei:10}
Mei, Y. (2010).
\newblock Efficient scalable schemes for monitoring a large number of data
  streams.
\newblock {\em Biometrika\/}~{\em 97\/}(2), 419--433.

\bibitem[\protect\citeauthoryear{Page}{Page}{1954}]{page1954continuous}
Page, E.~S. (1954).
\newblock Continuous inspection schemes.
\newblock {\em Biometrika\/}~{\em 41\/}(1/2), 100--115.

\bibitem[\protect\citeauthoryear{Perron}{Perron}{2006}]{perron2005dealing}
Perron, P. (2006).
\newblock Dealing with structural breaks.
\newblock {\em Palgrave Handbook of Econometrics Vol. 1: Econometric Theory, K.
  Patterson and T.C. Mills (eds.), Palgrave Macmillan\/}, 278--352.

\bibitem[\protect\citeauthoryear{Polunchenko and Tartakovsky}{Polunchenko and
  Tartakovsky}{2012}]{polunchenko2012state}
Polunchenko, A.~S. and A.~G. Tartakovsky (2012).
\newblock State-of-the-art in sequential change-point detection.
\newblock {\em Methodology and computing in applied probability\/}~{\em
  14\/}(3), 649--684.

\bibitem[\protect\citeauthoryear{Shao}{Shao}{2010}]{shao10}
Shao, X. (2010).
\newblock A self-normalized approach to confidence interval construction in
  time series.
\newblock {\em Journal of Royal Statistical Society, Series B\/}~{\em 72},
  343--366.

\bibitem[\protect\citeauthoryear{Shao}{Shao}{2015}]{shao15}
Shao, X. (2015).
\newblock Self-normalization for time series: a review of recent developments.
\newblock {\em Journal of the American Statistical Association\/}~{\em 110},
  1797--1817.

\bibitem[\protect\citeauthoryear{Shao and Wu}{Shao and Wu}{2007}]{shaowu2007}
Shao, X. and W.~B. Wu (2007).
\newblock Asymptotic spectral theory for nonlinear time series.
\newblock {\em The Annals of Statistics\/}~{\em 35\/}(4), 1773--1801.

\bibitem[\protect\citeauthoryear{Shao and Zhang}{Shao and
  Zhang}{2010}]{shaozhang10}
Shao, X. and X.~Zhang (2010).
\newblock Testing for change points in time series.
\newblock {\em Journal of the American Statistical Association\/}~{\em 105},
  1228--1240.

\bibitem[\protect\citeauthoryear{Wald}{Wald}{1945}]{wald1945}
Wald, A. (1945).
\newblock Sequential tests of statistical hypotheses.
\newblock {\em Annals of Mathematical Statistics\/}~{\em 16}, 117--186.

\bibitem[\protect\citeauthoryear{Wang and Shao}{Wang and
  Shao}{2020}]{wangshao20}
Wang, R. and X.~Shao (2020).
\newblock Dating the break in high-dimensional data.
\newblock {\em Available at https://arxiv.org/pdf/2002.04115.pdf\/}.

\bibitem[\protect\citeauthoryear{Wang, Volgushev, and Shao}{Wang
  et~al.}{2019}]{wang2019inference}
Wang, R., S.~Volgushev, and X.~Shao (2019).
\newblock Inference for change points in high dimensional data.
\newblock {\em arXiv preprint arXiv:1905.08446\/}.

\bibitem[\protect\citeauthoryear{Wang and Samworth}{Wang and
  Samworth}{2018}]{wang2018high}
Wang, T. and R.~J. Samworth (2018).
\newblock High dimensional change point estimation via sparse projection.
\newblock {\em Journal of the Royal Statistical Society: Series B (Statistical
  Methodology)\/}~{\em 80\/}(1), 57--83.

\bibitem[\protect\citeauthoryear{Wang and Mei}{Wang and Mei}{2015}]{wang:15}
Wang, Y. and Y.~Mei (2015).
\newblock Large-scale multi-stream quickest change detection via shrinkage
  post-change estimation.
\newblock {\em IEEE Transactions on Information Theory\/}~{\em 61\/}(12),
  6926--6938.

\bibitem[\protect\citeauthoryear{Wied and Galeano}{Wied and
  Galeano}{2013}]{wied2013monitoring}
Wied, D. and P.~Galeano (2013).
\newblock Monitoring correlation change in a sequence of random variables.
\newblock {\em Journal of Statistical Planning and Inference\/}~{\em 143\/}(1),
  186--196.

\bibitem[\protect\citeauthoryear{Wu}{Wu}{2005}]{wu2005nonlinear}
Wu, W.~B. (2005).
\newblock Nonlinear system theory: Another look at dependence.
\newblock {\em Proceedings of the National Academy of Sciences\/}~{\em
  102\/}(40), 14150--14154.

\bibitem[\protect\citeauthoryear{Wu and Shao}{Wu and Shao}{2004}]{wu2004limit}
Wu, W.~B. and X.~Shao (2004).
\newblock Limit theorems for iterated random functions.
\newblock {\em Journal of Applied Probability\/}~{\em 41\/}(2), 425--436.

\bibitem[\protect\citeauthoryear{Xie and Siegmund}{Xie and
  Siegmund}{2013}]{xie2013sequential}
Xie, Y. and D.~Siegmund (2013).
\newblock Sequential multi-sensor change-point detection.
\newblock {\em The Annals of Statistics\/}~{\em 41\/}(2), 670--692.

\bibitem[\protect\citeauthoryear{Yan, Paynabar, and Shi}{Yan
  et~al.}{2014}]{yan2014image}
Yan, H., K.~Paynabar, and J.~Shi (2014).
\newblock Image-based process monitoring using low-rank tensor decomposition.
\newblock {\em IEEE Transactions on Automation Science and Engineering\/}~{\em
  12\/}(1), 216--227.

\bibitem[\protect\citeauthoryear{Yan, Paynabar, and Shi}{Yan
  et~al.}{2018}]{yan2018real}
Yan, H., K.~Paynabar, and J.~Shi (2018).
\newblock Real-time monitoring of high-dimensional functional data streams via
  spatio-temporal smooth sparse decomposition.
\newblock {\em Technometrics\/}~{\em 60\/}(2), 181--197.

\bibitem[\protect\citeauthoryear{Yu and Chen}{Yu and Chen}{2017}]{yu2017finite}
Yu, M. and X.~Chen (2017).
\newblock Finite sample change point inference and identification for
  high-dimensional mean vectors.
\newblock {\em arXiv preprint arXiv:1711.08747\/}.

\bibitem[\protect\citeauthoryear{Yu and Chen}{Yu and Chen}{2019}]{yu2019robust}
Yu, M. and X.~Chen (2019).
\newblock A robust bootstrap change point test for high-dimensional location
  parameter.
\newblock {\em arXiv preprint arXiv:1904.03372\/}.

\bibitem[\protect\citeauthoryear{Zhang, Wang, and Shao}{Zhang
  et~al.}{2020}]{zhangwangshao20}
Zhang, Y., R.~Wang, and X.~Shao (2020).
\newblock Adaptive change point inference for high-dimensional data.
\newblock {\em Preprint\/}.

\bibitem[\protect\citeauthoryear{Zou, Wang, Zi, and Jiang}{Zou
  et~al.}{2015}]{zou2015efficient}
Zou, C., Z.~Wang, X.~Zi, and W.~Jiang (2015).
\newblock An efficient online monitoring method for high-dimensional data
  streams.
\newblock {\em Technometrics\/}~{\em 57\/}(3), 374--387.

\end{thebibliography}

\vskip .65cm
\noindent
Teng Wu, Department of Statistics, University of Illinois at Urbana-Champaign
\vskip 2pt
\noindent
E-mail: tengwu2@illinois.edu
\vskip 2pt

\noindent
Runmin Wang, Department of Statistical Science, Southern Methodist University
\vskip 2pt
\noindent
E-mail: runminw@smu.edu
\vskip 2pt
\noindent
Hao Yan, School of Computing Informatics \& Decision Systems Engineering, Arizona State University
\vskip 2pt
\noindent
E-mail:  haoyan@asu.edu
\vskip 2pt
\noindent
Xiaofeng Shao, Department of Statistics, University of Illinois at Urbana-Champaign
\vskip 2pt
\noindent
E-mail: xshao@illinois.edu
\end{document}